\documentclass[lettersize,journal]{IEEEtran}
\usepackage{amsmath,amsfonts,amssymb,amsthm}
\usepackage{array}
\usepackage{textcomp}
\usepackage{stfloats}
\usepackage{url}
\usepackage{verbatim}
\usepackage{graphicx}
\def\BibTeX{{\rm B\kern-.05em{\sc i\kern-.025em b}\kern-.08em
    T\kern-.1667em\lower.7ex\hbox{E}\kern-.125emX}}
\usepackage{balance}
\usepackage{cite}
\usepackage{xcolor}

\usepackage{caption}
\usepackage{subcaption}
\usepackage[labelformat=simple]{subcaption}

\newtheorem{theorem}{Theorem}
\newtheorem{lemma}{Lemma}
\newtheorem*{proof*}{Proof}
\newtheorem{remark}{Remark}
\newtheorem{proposition}{Proposition}

\usepackage{algorithm}
\usepackage[noend]{algpseudocode}

\usepackage{fancyhdr}

\usepackage{booktabs}
\usepackage{tablefootnote}

\begin{document}
\title{Sparse Array Beamformer Design via ADMM \\ (with updated Appendix D)}
\author{Huiping Huang, \textit{Student Member, IEEE}, Hing Cheung So, \textit{Fellow, IEEE}, and Abdelhak M. Zoubir, \textit{Fellow, IEEE} \vspace{-8mm}
\thanks{H. Huang and A. M. Zoubir are with Department of Electrical Engineering and Information Technology, Technische Universität Darmstadt, Germany (emails: h.huang@spg.tu-darmstadt.de, zoubir@spg.tu-darmstadt.de). H. C. So is with Department of Electrical Engineering, City University of Hong Kong, China (email: hcso@ee.cityu.edu.hk). The work of H. Huang is supported by the Graduate School CE within the Centre for Computational Engineering at Technische Universität Darmstadt. Parts of this paper were presented at SAM 2022 \cite{Huang2022June}.}}

\markboth{}
{Sparse Array Beamformer Design via ADMM}

\maketitle


\begin{abstract}
In this paper, we devise a sparse array design algorithm for adaptive beamforming. Our strategy is based on finding a sparse beamformer weight to maximize the output signal-to-interference-plus-noise ratio (SINR). The proposed method utilizes the alternating direction method of multipliers (ADMM), and admits closed-form solutions at each ADMM iteration. The algorithm convergence properties are analyzed by showing the monotonicity and boundedness of the augmented Lagrangian function. In addition, we prove that the proposed algorithm converges to the set of Karush-Kuhn-Tucker stationary points. Numerical results exhibit its excellent performance, which is comparable to that of the exhaustive search approach, slightly better than those of the state-of-the-art solvers, including the semidefinite relaxation (SDR), its variant (SDR-V), and the successive convex approximation (SCA) approaches, and significantly outperforms several other sparse array design strategies, in terms of output SINR. Moreover, the proposed ADMM algorithm outperforms the SDR, SDR-V, and SCA methods, in terms of computational complexity.
\end{abstract}


\section{Introduction}
\label{introduction}
\IEEEPARstart{A}{daptive} arrays have been widely applied in diverse practical applications, such as radar, sonar, wireless communications, to name just a few \cite{Gabriel1976}. One of their important functions is beamforming, which is to extract the signal-of-interest (SOI) while suppressing interference and noise \cite{Vorobyov2014}. It has been reported that the performance of beamforming is affected by not only the beamformer weight, but also the array configuration \cite{Lin1982}. In this sense, conventional uniform arrays may not be the optimal choices for adaptive beamformer design. On the other hand, since sparse arrays achieve increased array aperture and degrees of freedom while reducing the hardware complexity, as compared to conventional uniform arrays. Thus, sparse arrays could be a better option for adaptive beamformer design.

On the other hand, many research works have considered beamformer design for the communication systems of hybrid structures \cite{Zhang2005, Sudarshan2006, Venkateswaran2010, Heath2016}. Hybrid architectures are one approach for providing enhanced benefits of multiple-input multiple-output (MIMO) communication at millimeter-wave (mmWave) frequencies. In such a hybrid system, there are smaller number of radio-frequency (RF) chains than antennas. Note that RF chains are much more expensive than antennas. Hybrid structure offers a compromise between system performance and hardware complexity. For such a configuration, it raises a question on how to adaptively switch the available RF chains to the corresponding subset of antennas, which can be interpreted as sparse array beamformer design. 

\subsection{Related Works}

In recent years, several strategies for designing sparse array beamformers have been proposed, see for example \cite{Hawes2012, Zhou2018, Shi2019, Shi2014, Wang2014, Wang2021, Wang2023, Shrestha2023, Hamza2019, Hamza2021} and references therein. \cite{Hawes2012} investigates sparse array design problem in the presence of steering vector mismatch. \cite{Zhou2018} focuses on a specific sparse array, i.e., coprime array, for adaptive beamforming. A constrained normalized least-mean-square approach is developed in \cite{Shi2019}. However, there is no theoretical analyses of the convergence of the algorithm. The authors in \cite{Shi2014} consider the problem of resource allocation and interference management, and propose a group sparse beamforming method. A reconfigurable adaptive antenna array strategy is proposed to minimize the spatial correlation coefficient between the desired signal and the interference \cite{Wang2014}. \cite{Wang2021} studies the sparse array beamforming design problem by firstly dividing the whole array into several non-overloped sub-groups, and then selecting one antenna per sub-group. By using the regularized switching network, \cite{Wang2023} proposes a cognitive sparse beamformer, which is adaptive to the environmental dynamics. In \cite{Shrestha2023}, the authors investigate joint beamforming and antenna selection problem under imperfect channel state information, and develop an efficient branch and bound (B\&B) approach and a machine learning based scheme. \cite{Hamza2019} and \cite{Hamza2021} consider sparse array design for achieving maximum signal-to-interference-plus-noise ratio (SINR) for different scenarios of single point source, multiple point sources, and wideband sources. Different from most of the existing works, our goal is to maximize the output SINR by using a smaller number of RF chains (compared to the number of antennas), and our main contributions include developing an efficient algorithm to solve the resulting problem and providing theoretical guarantee of the convergence property of the proposed algorithm.

From the perspective of algorithm, existing methods can be roughly divided into three categories: Greedy based \cite{Shi2014, Wang2014, Hegde2018}, machine learning based \cite{Elbir2020, Diamantaras2021, Xu2021, Vu2021, Shrestha2023}, and optimization based \cite{Shrestha2023, Cheng2013, Fischer2018, Li2018, Nai2010, Mehanna2012, Mehanna2013, Gkizeli2014, Hawes2014, Demir2015, Zheng2020, Hamza2019, Hamza2021, Zhai2017, Arora2021, Zheng2021, Yu2019, Chen2021, Dan2022} approaches. The greedy procedure in \cite{Shi2014, Wang2014, Hegde2018} largely reduces the combinatorial exploration space, but could result in a highly suboptimal solution. Machine learning techniques \cite{Elbir2020, Diamantaras2021, Xu2021, Vu2021, Shrestha2023} require prior data for their training step, which might be unavailable in some practical scenarios.    
On the other hand, optimization based methods include B\&B \cite{Shrestha2023}, mixed-integer programming (MIP) \cite{Cheng2013, Fischer2018, Li2018}, semidefinite relaxation (SDR) \cite{Mehanna2012, Mehanna2013, Hamza2019, Hamza2021, Zheng2020, Dan2022}, and successive convex approximation (SCA) \cite{Zhai2017, Arora2021, Zheng2021}. B\&B and MIP are capable of finding the global optimum at the cost of a large computational burden. SDR and SCA based methods are computationally expensive when the dimension of the resulting matrix is high \cite{Mehanna2015, Huang2016}. Besides, the relaxation nature of SDR usually leads to a solution with a rank not being one, in which case extra post-processing based on randomization is needed \cite{Sidiropoulos2006}. Moreover, note that the SDR methods in \cite{Mehanna2012, Mehanna2013, Hamza2019, Hamza2021} utilize the $\ell_{1}$-norm square instead of the $\ell_{1}$-norm for sparsity promotion. Although the simulation results in these papers and also in Section \ref{beamformingperformance} of the present paper demonstrate the usefulness of the SDR-type methods, no theoretical support is available due to the convexity of the Pareto boundary that is not guaranteed (which has also been mentioned in \cite{Mehanna2012} and \cite{Mehanna2013}). On the other hand, another downside of the SCA approach lies in the fact that it requires a feasible starting point, which could be a difficult task on its own \cite{Mehanna2015}.

\subsection{Contributions}

In this paper, a sparse array design algorithm based on the alternating direction method of multipliers (ADMM) is devised for adaptive beamforming. The proposed technique admits closed-form solutions at each ADMM iteration. Convergence analyses of the proposed algorithm are provided by showing the monotonicity and boundedness of the augmented Lagrangian function. Additionally, it is proved that the proposed algorithm converges to the set of Karush-Kuhn-Tucker (KKT) stationary points. Simulation results demonstrate excellent behavior of our scheme, as it outperforms several existing methods, and is comparable to the exhaustive search approach.

Note that parts of this paper were presented in its conference precursor, i.e., \cite{Huang2022June}. Compared to \cite{Huang2022June}, the advancements and supplementary contributions made in the present article are listed as follows.
\begin{itemize}
\item In \cite{Huang2022June}, we only provide an illustrative example for Proposition \ref{theorem_closest} without a proof. Differently, in the present article, besides an illustrative example we also provide a mathematical proof in a rigorous manner.
\item The convergence property of the proposed ADMM has not been analysed in \cite{Huang2022June}. While in the present article, we provide such convergence analyses in theory, and extensive numerical simulations are provided to validate the theoretical development, see Sections \ref{convergence_ADMM} and \ref{simulation} in the present article, respectively.
\item In the present article we consider the computational complexity of the algorithms, see Section \ref{computationalcost}, which is not mentioned in \cite{Huang2022June}.
\item Moreover, more existing sparse array design strategies (such as enumeration method, random sparse array, compact uniform linear array, nested array, co-prime array, semidefinite relaxation approach, and successive convex approximation method) are taken into account in the present article.
\end{itemize}

On the other hand, our algorithms and theoretical results are developed primarily on the basis of the ideas presented in \cite{Huang2016} and \cite{Hong2016}. Several differences are highlighted as follows.
\begin{itemize}
\item Different from \cite{Huang2016} which used ADMM to solve general quadratically constrained quadratic programming (QCQP) problems, we focus on a specific QCQP problem that arises in sparse array beamformer design. Our problem involves an $\ell_{1}$-norm regularization and thus our solution is sparse, which is not the case in \cite{Huang2016}. 

\item Another important difference between our work and \cite{Huang2016} lies in the fact that the latter only provides a weaker convergence result, i.e., \textit{if ADMM converges for their problem, then} it converges to a KKT stationary point, see Theorem 1 in \cite{Huang2016}. On the other hand, we show stronger convergence results for our algorithm. That is, we first prove the convergence of the proposed algorithm under a mild condition, and then we prove that it converges to a KKT stationary point.

\item Note that \cite{Hong2016} needed extra assumptions on the Lipschitz gradient continuity as well as the boundedness of their objective function. In our work, we require neither such assumptions nor any other assumptions.

\item Since \cite{Hong2016} considered general non-convex problems, no explicit expressions for their parameters were derived. On the contrary, as we consider a specific non-convex problem of sparse array beamformer design, the properties of our objective function have been investigated and thus several parameters are given in an explicit manner. See for example the augmented Lagrangian parameter $\rho$ and the strongly convex parameter $\gamma_{{\bf v}}$ in Lemma \ref{lem_monotonicity}.


\item The results in \cite{Hong2016} were based on the augmented Lagrangian function, while we exploit the scaled-form augmented Lagrangian function. This results in significant differences in the following three aspects: i) the proof of the monotonicity of the augmented Lagrangian function, ii) the proof of the property of the point sequence, and iii) the proof that the algorithm converges to the set of KKT stationary points; see the proofs of Lemma \ref{lem_monotonicity}, Theorem \ref{theorem_convergence}, and Theorem \ref{theorem_limit_stationary}, respectively.     
\end{itemize}

\subsection{Paper Structure}

The rest of the paper is organized as follows. The signal model is established in Section \ref{SignalModel}. The proposed approach is presented in Section \ref{proposed} and the convergence analyses are given in Section \ref{convergence_ADMM}. ADMM with re-weighted $\ell_{1}$-norm regularization is proposed in Section \ref{ADMM_enhanced_sparsity}. Section \ref{simulation} shows the simulation results, and Section \ref{conclusion} concludes the paper and provides possible future works.

\subsection{Notation}

The notation used throughout the paper is summarised in Table \ref{Table_notations}.

\begin{table*}[t]
\caption{Summary of Notation}
\label{Table_notations}
\centering{}%
\begin{tabular}{cc|cc}
\toprule 
Notation & Meaning  & Notation & Meaning \tabularnewline
\midrule
\midrule 
regular letter & scalar &  $\|\cdot\|_{0}$  & $\ell_{0}$-quasi-norm of a vector  \tabularnewline
bold-faced lower-case letter & vector  &  $\|\cdot\|_{1}$  & $\ell_{1}$-norm of a vector    \tabularnewline
bold-faced upper-case letter  & matrix  & $\|\cdot\|_{2}$  & $\ell_{2}$-norm of a vector    \tabularnewline
${\bf I}$ & identity matrix  &  $| \cdot |$  & absolute value  \tabularnewline
${\bf 1}_{M}$ & $M \! \times \! M$ all-one matrix  & $(\cdot)_{{\bold +}}$  & $(x)_{{\bold +}} \triangleq \max\{x, 0\}$  \tabularnewline
${\bf 1}$ & all-one vector   &  $\max\{a, b\}$  & maximum value of $a$ and $b$  \tabularnewline
${\bf 0}$ & all-zero vector   &  $\langle {\bf x} , {\bf y}\rangle$  & inner product, $\langle {\bf x} , {\bf y}\rangle = {\bf x}^{\mathrm{H}}{\bf y}$  \tabularnewline
$\cdot^{\mathrm{T}}$  & transpose  &  $\text{Tr}\{\cdot\}$  & matrix trace  \tabularnewline
$\cdot^{\mathrm{H}}$  & Hermitian transpose  & $\text{E}\{\cdot\}$  &  mathematical expectation  \tabularnewline
$\cdot^{-1}$  & matrix inverse   &  ${\bf X} \succeq 0$  &  ${\bf X}$ is positive semidefinite   \tabularnewline
$\lambda_{\mathrm{max}}(\cdot)$ & largest eigenvalue  & ${\bf X} \succeq {\bf Y}$  &  ${\bf X} - {\bf Y}$ is positive semidefinite   \tabularnewline
$\lambda_{\mathrm{min}}(\cdot)$ &  smallest eigenvalue  &  $\geq$  &  element-wise larger than or equal to  \tabularnewline
$\mathbb{C}$  & set of complex numbers  &  $\odot$  & element-wise multiplication     \tabularnewline
$\Re\{\cdot\}$  & real part   &   $\oslash$  & element-wise division    \tabularnewline
$\jmath$ & imaginary unit, $\jmath = \sqrt{-1}$  &  $\text{sign}(\cdot)$  &    $\text{sign}(x) = \left\{ \! \begin{array}{ll}
~0, &  \text{if}~x = 0, \\
\frac{x}{|x|}, &  \text{otherwise}
  \end{array}  \right.$  \tabularnewline 
\bottomrule
\end{tabular}
\end{table*}

\section{Signal Model}
\label{SignalModel}
Consider a compact uniform linear array (ULA) consisting of $M$ antenna sensors, where the terminology \textit{compact} means that the element-spacing of two adjacent antennas is equal to half-wavelength of the incident signals. We refer the ULA with element-spacing larger than half-wavelength to as \textit{sparse ULA}. Denote ${\bf x}(t) \in \mathbb{C}^{M}$ as the observation vector of the compact ULA, which can be modeled as:
\begin{align}
{\bf x}(t) = {\bf a}(\theta_{0})s_{0}(t) + \sum_{k = 1}^{K}{\bf a}(\theta_{k})s_{k}(t) + {\bf n}(t),
\end{align}
where $t = 1, 2, \cdots, T$ denotes the time index, with $T$ being the total number of available snapshots, $\theta_{0}$ and $s_{0}(t)$ are the direction-of-arrival (DOA) and waveform of the SOI, respectively, while $\theta_{k}$ and $s_{k}(t)$ denote those of the $k$-th interference signal. We consider one SOI and $K$ unknown interferers in the data model, and the SOI and interferers are assumed to be narrow-band, far-field, and uncorrelated signals. In addition, ${\bf n}(t) \in \mathbb{C}^{M}$ is a zero-mean Gaussian noise vector, and the array steering vector ${\bf a}(\theta) \in \mathbb{C}^{M}$ takes the form as 
\begin{align}
{\bf a}(\theta) = [1, e^{- \jmath \pi \sin(\theta)}, \cdots, e^{- \jmath \pi (M - 1) \sin(\theta)}]^{\mathrm{T}}.
\end{align}
For the simplicity of notation, we denote ${\bf a}(\theta_{k})$ as ${\bf a}_{k}$, for all $k = 0, 1, \cdots, K$, here and subsequently.

The beamformer output is calculated as
\begin{align}
y(t) = {\bf w}^{\mathrm{H}}{\bf x}(t),
\end{align}
in which ${\bf w} \in \mathbb{C}^{M}$ is the beamformer weight vector to be designed. The beamformer output SINR is defined as \cite{Vorobyov2014}
\begin{align}
\label{sinr}
\text{SINR} = \frac{\sigma_{s}^{2}|{\bf w}^{\mathrm{H}}{\bf a}_{0}|^{2}}{{\bf w}^{\mathrm{H}}{\bf R}_{{i+n}}{\bf w}},
\end{align}
where $\sigma_{s}^{2} = \text{E}\{|s_{0}(t)|^{2}\}$ is the power of the SOI, and ${\bf R}_{{i+n}}$ is the interference-plus-noise covariance matrix, which can be written as
\begin{align}
\label{R_in}
{\bf R}_{{i+n}} = \sum_{k = 1}^{K}\sigma_{k}^{2}{\bf a}_{k}{\bf a}_{k}^{\mathrm{H}} + \sigma_{n}^{2}{\bf I},
\end{align}
assuming that the interference signals are uncorrelated with the noise. In (\ref{R_in}), $\sigma_{k}^{2} = \text{E}\{|s_{k}(t)|^{2}\}$ is the power of the $k$-th interference signal, and $\sigma_{n}^{2}$ is the noise power.

One of the most prevailing strategies for beamformer design is to maximize the SINR, which leads to the minimum variance distortionless response (MVDR) beamformer design \cite{Capon1969}:
\begin{align}
\label{MVDR}
\min_{{\bf w}} ~ {\bf w}^{\mathrm{H}}{\bf R}_{i+n}{\bf w} \quad \text{subject to} ~ |{\bf w}^{\mathrm{H}}{\bf a}_{0}|^{2} = 1.
\end{align}
The above problem can be reformulated equivalently by replacing the realistically unattainable ${\bf R}_{i+n}$ with the received data covariance matrix ${\bf R}_{x} = \sigma_{s}^{2}{\bf a}_{0}{\bf a}_{0}^{\mathrm{H}} + {\bf R}_{i+n}$ \cite{Shahbazpanahi2003}, as:
\begin{align}
\label{SMI}
\min_{{\bf w}} ~ {\bf w}^{\mathrm{H}}{\bf R}_{x}{\bf w} \quad \text{subject to} ~ |{\bf w}^{\mathrm{H}}{\bf a}_{0}|^{2} \geq 1.
\end{align}
It is worth noting that the equality constraint is formulated as an inequality one, since the output power of the SOI is included as part of the objective function in (\ref{SMI}) \cite{Hamza2019, Hamza2021}.

The above problems have a closed-form solution as ${\bf w}_{\textrm{opt}} = \frac{{\bf R}_{i+n}^{-1}{\bf a}_{0}}{{\bf a}_{0}^{\textrm{H}}{\bf R}_{i+n}^{-1}{\bf a}_{0}}$. Substituting ${\bf w}_{\mathrm{opt}}$ into (\ref{sinr}) yields the corresponding optimum output SINR as      
\begin{align}
\textrm{SINR}_{\textrm{opt}} = \sigma_{s}^{2}{\bf a}_{0}^{\textrm{H}}{\bf R}_{i+n}^{-1}{\bf a}_{0}.
\end{align}

\section{Sparse Array Beamformer Design}
\label{proposed}
\subsection{Sparse Beamforming Problem}
Consider the situation where only $L \leq M$ RF chains are available \cite{Mehanna2012}, and thus only $L$ antennas can be simultaneously utilized for beamformer design. The problem can be formulated as \cite{Mehanna2012, Mehanna2013, Hamza2019, Hamza2021}:
\begin{align}
\label{L0_norm}
\min_{{\bf w}} ~ {\bf w}^{\mathrm{H}}{\bf R}_{x}{\bf w}  ~~ \text{subject to} ~ |{\bf w}^{\mathrm{H}}{\bf a}_{0}|^{2} \geq 1 \text{~and~} \|{\bf w}\|_{0} = L.
\end{align}
This is a combinatorial problem, and there are $M \choose L$ possible options. It could be an extremely huge number when $M$ is large and $L$ is moderate. For instance, if $M = 100$ and $L = 20$, there are totally $100 \choose 20$ $> 5 \times 10^{20}$ subproblems \cite{Diamantaras2021}. Even if a modern machine (as fast as $10^{-10}$ seconds per subproblem) is adopted, it still needs more than 1.5 thousand years in total. This is evidently unacceptable, and thus more computationally efficient approaches are required.

One widespread method is to replace the non-convex constraint $\|{\bf w}\|_{0} = L$ with its convex surrogates, such as the $\ell_{1}$-norm. By doing so and writing the $\ell_{1}$-norm in the objective function as a penalty, we relax Problem (\ref{L0_norm}) as \cite{Mehanna2012, Zheng2020, Hamza2019}
\begin{align}
\label{L1_norm}
\min_{{\bf w}} ~ {\bf w}^{\mathrm{H}}{\bf R}_{x}{\bf w} + \lambda\|{\bf w}\|_{1}  \quad \text{subject to} ~ |{\bf w}^{\mathrm{H}}{\bf a}_{0}|^{2} \geq 1,
\end{align}
where $\lambda > 0$ is a tuning parameter controlling the sparsity of the solution (i.e., the number of selected sensors). The above problem is QCQP \cite{Park2017} with $\ell_{1}$-regularization, and it is still non-convex because of its constraint. State-of-the-art solvers include SDR, SCA, ADMM, and their variants, see \cite{Park2017, Mehanna2015, Huang2023, Boyd2004, Huang2016, Boyd2011, Luo2010} for general QCQP problems and \cite{Mehanna2012, Mehanna2013, Zheng2020, Hamza2019, Hamza2021, Zhai2017, Dan2022, Cheng2021, Cheng2022, Liu2018, Wei2021, Chen2017} for specific QCQP problems with applications in MIMO radar, wireless communications, and so on. In the following subsection, we develop a method based on ADMM for solving Problem (\ref{L1_norm}), which shall be shown to have closed-form solutions at each iteration.

\subsection{ADMM for Problem (\ref{L1_norm})}
\label{admm_prob_L1norm}
To solve Problem (\ref{L1_norm}) using ADMM, we first introduce an auxiliary variable ${\bf v} \in \mathbb{C}^{M}$ and reformulate (\ref{L1_norm}) into   
\begin{subequations}
\label{ADMM_auxility}
\begin{align}
\label{ADMM_obj}
\min_{{\bf w}, {\bf v}} ~~~~ & {\bf v}^{\mathrm{H}}{\bf R}_{x}{\bf v} + \lambda\|{\bf w}\|_{1}  \\
\label{ADMM_st}
\text{subject to} ~ & \left\{ \begin{array}{l}
 |{\bf w}^{\mathrm{H}}{\bf a}_{0}|^{2} \geq 1 \\
 {\bf w} = {\bf v}.
 \end{array}
 \right.
\end{align}
\end{subequations}
Then we can write down the scaled-form ADMM iterations for Problem (\ref{L1_norm}) as \cite{Boyd2011} 
\begin{subequations}
\label{ADMM_iterate}
\begin{align}
\label{ADMM_w}
{\bf w} ~ & \gets ~ \left\{ \begin{array}{l}
\displaystyle\arg\min_{{\bf w}} ~ \lambda\|{\bf w}\|_{1} + \frac{\rho}{2}\|{\bf w} - {\bf v} + {\bf u}\|_{2}^{2} \\ 
\text{subject to} ~ |{\bf w}^{\mathrm{H}}{\bf a}_{0}|^{2} \geq 1
\end{array}
\right. \\
\label{ADMM_v}
{\bf v}  ~ & \gets ~ \arg\min_{{\bf v}} ~ {\bf v}^{\mathrm{H}}{\bf R}_{x}{\bf v} + \frac{\rho}{2}\|{\bf w} - {\bf v} + {\bf u}\|_{2}^{2} \\
\label{ADMM_u}
{\bf u} ~ & \gets ~ {\bf u} + {\bf w} - {\bf v}
\end{align}
\end{subequations}
where the original variable ${\bf w}$ and the auxiliary variable ${\bf v}$ are separately treated in (\ref{ADMM_w}) and (\ref{ADMM_v}), respectively, ${\bf u}$ is the scaled dual variable corresponding to the equality constraint in (\ref{ADMM_st}), i.e., ${\bf w} = {\bf v}$, and $\rho > 0$ is the augmented Lagrangian parameter.

In what follows, we show that (\ref{ADMM_iterate}) has closed-form solutions at each ADMM iteration, by deducing ${\bf w}$ and ${\bf v}$ from (\ref{ADMM_w}) and (\ref{ADMM_v}), respectively. First of all, from (\ref{ADMM_v}), it is simple to arrive at the closed-form solution of ${\bf w}$ in a least-squares form, as
\begin{align}
\label{solution_v}
{\bf v} = \rho(2{\bf R}_{x} + \rho{\bf I})^{-1}({\bf w} + {\bf u}).
\end{align}

Now we turn to (\ref{ADMM_w}). We solve (\ref{ADMM_w}) in two steps: i) We consider the unconstrained minimization problem by directly removing its constraint; ii) We check whether the solution obtained from Step i) satisfies the constraint, and update the final solution accordingly. Details are provided as follows.

\textbf{Step i):} We consider the following unconstrained minimization problem 
\begin{align}
\label{v_unconstrained}
\min_{{\bf w}} ~ \lambda\|{\bf w}\|_{1} + \frac{\rho}{2}\|{\bf w} - {\bf v} + {\bf u}\|_{2}^{2}.
\end{align}
By calculating the subgradient of the objective function in (\ref{v_unconstrained}) with respect to (w.r.t.) ${\bf w}$, and setting the resultant expression equal to zero, we obtain its solution, denoted by ${\bf{\bar{w}}}$, as 
\begin{align}
\label{w_bar}
{\bf{\bar{w}}} = \text{sign}({\bf v} - {\bf u}) \odot \left(|{\bf v} - {\bf u}| - \frac{\lambda}{\rho}\right)_{{\bold +}}.
\end{align}
The detailed derivation of (\ref{w_bar}) from Problem (\ref{v_unconstrained}) is omitted here, and the interested readers are referred to the similar result in Lemma 1 in \cite{Zoubir2018}.

\textbf{Step ii):} We check whether or not ${\bf{\bar{w}}}$ obtained from (\ref{w_bar}) statisfies $|{\bf{\bar{w}}}^{\mathrm{H}}{\bf a}_{0}|^{2} \geq 1$. If it is, then the solution to (\ref{ADMM_w}), referred to as ${\bf{\widehat{w}}}$, is ${\bf{\widehat{w}}} = {\bf{\bar{w}}}$. If it is not, then ${\bf{\widehat{w}}}$ can be found via the following theorem.

\begin{proposition}
\label{theorem_closest}
Denote ${\bf{\widehat{w}}}$ and ${\bf{\bar{w}}}$ as the solutions to Problems (\ref{ADMM_w}) and (\ref{v_unconstrained}), respectively. If ${\bf{\bar{w}}} $ does not satisfy $|{\bf{\bar{w}}}^{\mathrm{H}}{\bf a}_{0}|^{2} \! \geq \! 1$, then as long as $\rho \gg \lambda$, ${\bf{\widehat{w}}}$ equals the one in $\{{\bf w} \! : \! |{\bf w}^{\mathrm{H}}{\bf a}_{0}|^{2} \! \geq \! 1 \}$, such that it is closest (in an $\ell_{2}$-norm sense) to ${\bf{\bar{w}}}$. 
\end{proposition}

\begin{proof}
Define $f({\bf{w}}) \triangleq \lambda\|{\bf w}\|_{1} + \frac{\rho}{2}\|{\bf w} - {\bf v} + {\bf u}\|_{2}^{2}$, and define ${\bf{\widetilde{w}}}$ such that   
\begin{align}
\label{w_widetilde}
\|{\bf{\widetilde{w}}} - {\bf{\bar{w}}}\|_{2} \leq \|{\bf w} - {\bf{\bar{w}}}\|_{2}
\end{align}
holds for any ${\bf w} \in \{{\bf w} : |{\bf w}^{\mathrm{H}}{\bf a}_{0}|^{2} \geq 1 \}$.   
The Lagrangian parameter is chosen as $\rho = C\lambda$, where $C$ is a constant. As shall be shown later, the augmented Lagrangian parameter $\rho$ is set to be large in order to make our algorithm converge. When $C \to \infty$ (i.e., $\rho \gg \lambda$), the objective function 
\begin{align}
\label{f_w_approx}
f({\bf w}) = \frac{\rho}{2}\|{\bf w} - {\bf v} + {\bf u}\|_{2}^{2} = \frac{\rho}{2}\|{\bf w} - {\bf{\bar{w}}}\|_{2}^{2},
\end{align}
where, in the second equality, we used ${\bf{\bar{w}}} = {\bf v} - {\bf u}$ as $C \to \infty$. 

Suppose that there exists a point ${\bf{w'}} \in \{{\bf w} : |{\bf w}^{\mathrm{H}}{\bf a}_{0}|^{2} \geq 1 \}$, such that $f({\bf{w'}}) < f({\bf{\widetilde{w}}})$. Thus, by using (\ref{f_w_approx}), we obtain that $\frac{\rho}{2}\|{\bf{w'}} - {\bf{\bar{w}}}\|_{2}^{2} < \frac{\rho}{2}\|{\bf{\widetilde{w}}} - {\bf{\bar{w}}}\|_{2}^{2}$, which contradicts (\ref{w_widetilde}). This implies that $f({\bf{\widetilde{w}}}) \leq f({\bf w})$ holds for all feasible ${\bf w}$, that is, ${\bf{\widetilde w}}$ is the solution to Problem (\ref{ADMM_w}).
\end{proof}

\begin{remark}
Besides the above mathematical proof, we give an illustrative example for Proposition \ref{theorem_closest}, by showing the near-symmetric structure of the objective function around its stationary point\footnote{We say a function $f({\bf x})$ has a \textit{near-symmetric} structure around a point ${\bf x}_{0}$ if and only if $f({\bf x}_{0} + {\bf x}) \approx f({\bf x}_{0} - {\bf x})$ holds for any ${\bf x}$.}. For simplicity, the variable is set to be real-valued and the dimension $M = 1$. The parameters are $\lambda = 1$, $\rho = 4$, ${\bf a}_{0} = 1/2$, and $-{\bf v} + {\bf u} = -1$. Problem (\ref{ADMM_w}) becomes
\begin{align}
\label{illustractive_example}
\min_{w} ~ f(w) \triangleq |w| + 2(w - 1)^{2} \quad \text{subject to} ~ |w| \geq 2,
\end{align}
with its stationary point $w_{0}$ falling outside its feasible region $|w| \geq 2$, as in Fig. \ref{f_w}. Thanks to the convexity and near-symmetric structure of the objective function $f(w)$, finding its minimum is equivalent to determining the point (inside the feasible region) closest to its stationary point $w_{0}$. In Fig. \ref{f_w}, it is easy to see that $w = 2$ is such a point, and thus it is the solution to Problem (\ref{illustractive_example}). 

\begin{figure}[t]
	\vspace*{-6mm}
	\centerline{\includegraphics[width=0.5\textwidth]{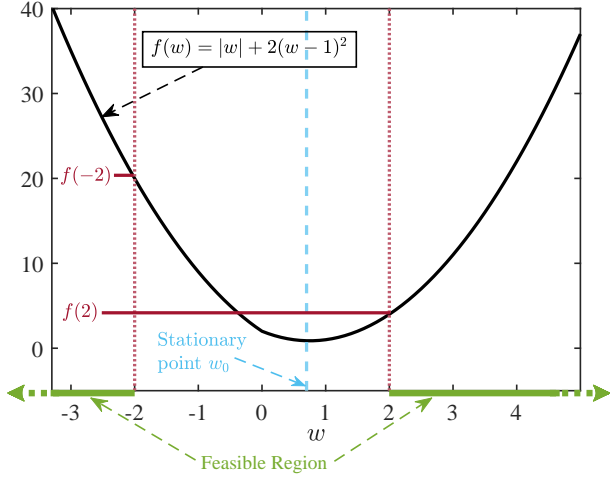}}
	\caption{Illustration of near-symmetric structure of $f(w)$.}
	\label{f_w}
\end{figure}

\end{remark}

Consequently, if ${\bf{\bar{w}}}$ obtained from (\ref{w_bar}) does not satisfy $|{\bf{\bar{w}}}^{\mathrm{H}}{\bf a}_{0}|^{2} \geq 1$, then according to Proposition \ref{theorem_closest}, ${\bf{\widehat{w}}}$ can be found by solving the following minimization problem
\begin{align}
\label{w_har_inequality}
{\bf{\widehat{w}}} ~ \gets ~ \arg\min_{{\bf w}} ~ \|{\bf w} - {\bf{\bar{w}}}\|_{2}^{2} \quad \text{subject to} ~ |{\bf w}^{\mathrm{H}}{\bf a}_{0}|^{2} \geq 1,
\end{align}
which is equivalent to 
\begin{align}
\label{w_har_equality}
{\bf{\widehat{w}}} ~ \gets ~ \arg\min_{{\bf w}} ~ \|{\bf w} - {\bf{\bar{w}}}\|_{2}^{2} \quad \text{subject to} ~ |{\bf w}^{\mathrm{H}}{\bf a}_{0}|^{2} = 1.
\end{align}
The equivalence between Problems (\ref{w_har_inequality}) and (\ref{w_har_equality}) is shown in Appendix \ref{proof_equiv}, and it indicates that the solutions to Problem (\ref{w_har_inequality}) always fall on the \textit{boundary} of its feasible region. Problem (\ref{w_har_equality}) has the following closed-form solution as \cite{Huang2016}
\begin{align}
{\bf{\widehat{w}}} = {\bf{\bar{w}}} + \frac{1 - |{\bf{\bar{w}}}^{\mathrm{H}}{\bf a}_{0}|}{\|{\bf a}_{0}\|_{2}^{2} |{\bf{\bar{w}}}^{\mathrm{H}}{\bf a}_{0}|} {\bf a}_{0}{\bf{\bar{w}}}^{\mathrm{H}}{\bf a}_{0}.
\end{align}

Eventually, by considering Steps i) and ii) simultaneously and making use of the plus function, the solution to (\ref{ADMM_w}) can be written in a single formula as
\begin{align}
\label{w_widehat}
{\bf{\widehat{w}}} = {\bf{\bar{w}}} + \frac{ \left( 1 - |{\bf{\bar{w}}}^{\mathrm{H}}{\bf a}_{0}| \right)_{{\bold +}} }{\|{\bf a}_{0}\|_{2}^{2} |{\bf{\bar{w}}}^{\mathrm{H}}{\bf a}_{0}|} {\bf a}_{0}{\bf{\bar{w}}}^{\mathrm{H}}{\bf a}_{0},
\end{align}
where ${\bf{\bar{w}}}$ is given in (\ref{w_bar}).

So far, we have derived closed-form solutions for ${\bf w}$ and ${\bf v}$ at each ADMM iteration. The complete ADMM for solving Problem (\ref{L1_norm}) is summarized in Algorithm \ref{admm_alg}, where $k_{\mathrm{max}}$ denotes a large scalar and $\eta$ a small one, used to terminate the iteration, and subscript $\cdot_{(k)}$ denotes the variable at the $k$-th iteration. The convergence property of the proposed ADMM algorithm will be discussed in Section \ref{convergence_ADMM}. 


\begin{remark}
Note that we can attain any level of sparsity (i.e., any number $L$ out of $M$ sensors in sparse array design), by carefully tuning the value of $\lambda$. This shall be verified in the simulation section, see Fig. \ref{sparsityofweight} in Section \ref{simulation}.
\end{remark}

\begin{remark}
\label{remark_L}
To ensure selection of $L$ sensors, an appropriate value of $\lambda$ is typically found by carrying out a binary search over a probable interval of $\lambda$, say $[\lambda_{L} , \lambda_{U}]$. To be precise, we begin by solving ${\bf{\widehat{w}}}$ using Algorithm \ref{admm_alg}, with $\lambda = (\lambda_{L} + \lambda_{U})/2$. Denote $\widehat{L}$ the number of entries of ${\bf{\widehat{w}}}$ larger than (in modulus) $\alpha$ times the maximum (in modulus) entry of ${\bf{\widehat{w}}}$. If ${\widehat{L}} > L$ (resp. ${\widehat{L}} < L$), then we update $\lambda_{L} = \lambda$ (resp. $\lambda_{U} = \lambda$), and solve another ${\bf{\widehat{w}}}$ with $\lambda = (\lambda_{L} + \lambda_{U})/2$. We repeat the above steps until ${\widehat{L}} = L$, and the sensors corresponding to the $L$ largest entries (in modulus) are selected. We set $\alpha = 0.1$ in our simulations. 
\end{remark}

\begin{remark}
Note that the solution of (\ref{L1_norm}) is not exactly equal to the one of (\ref{L0_norm}). Therefore, after the solution of desired sparsity of (\ref{L1_norm}) is obtained, one should solve a reduced-size minimization problem similar to (\ref{SMI}) as a last step, using only the selected sensors.
\end{remark}

\begin{algorithm}[t]
	\caption{ADMM for solving Problem (\ref{L1_norm})}
	\label{admm_alg}
	\textbf{Input~~~~\!:} ${\bf R}_{x} \in \mathbb{C}^{M \times M}$, ${\bf a}_{0} \in \mathbb{C}^{M}$, $\lambda$, $\rho$, $k_{\mathrm{max}}$, $\eta$ \\
	\textbf{Output~~\!:} ${\bf{\widehat w}} \in \mathbb{C}^{M}$ \\
	\textbf{Initialize:} ${\bf v}_{(0)} \gets {\bf v}_{\text{init}}$, ${\bf u}_{(0)} \gets {\bf u}_{\text{init}}$, $k \gets 0$
	\begin{algorithmic}[1]
		\While {not converged} \\
		 ~~ ${\bf{\bar{w}}}_{(k+1)} \gets  \text{sign}( \! {\bf v}_{(k)} - {\bf u}_{(k)} \! ) \! \odot \! \left( \! |{\bf v}_{(k)} \! - \! {\bf u}_{(k)}| \! - \! \frac{\lambda}{\rho} \! \right)_{\!\!{\bold +}} \!\! $ \\
		 ~~ ${\bf w}_{(k+1)} \gets {\bf{\bar{w}}}_{(k+1)} + \frac{\left( 1 - |{\bf{\bar{w}}}_{(k+1)}^{\mathrm{H}}{\bf a}_{0}| \right)_{{\bold +}}}{\|{\bf a}_{0}\|_{2}^{2}|{\bf{\bar{w}}}_{(k+1)}^{\mathrm{H}}{\bf a}_{0}|} {\bf a}_{0}{\bf{\bar{w}}}_{(k+1)}^{\mathrm{H}}{\bf a}_{0} $ \\
		 ~~ ${\bf v}_{(k+1)} \gets \rho(2{\bf R}_{x} + \rho{\bf I})^{-1}({\bf w}_{(k+1)} + {\bf u}_{(k)})$ \\
		 ~~ ${\bf u}_{(k+1)} \gets {\bf u}_{(k)} + {\bf w}_{(k+1)} - {\bf v}_{(k+1)}$ \\
		~~ converged $\gets$ $ k \! + \! 1 \geq k_{\mathrm{max}} $ or $ \| {\bf w}_{(k+1)} \! - \! {\bf v}_{(k+1)} \|_{2} \leq \eta $ \\
		 ~~ $k \gets k+1$
		\EndWhile \State \textbf{end while} \\
		${\bf{\widehat w}} \gets {\bf w}_{(k)}$
	\end{algorithmic}
\end{algorithm}

\section{Convergence Analysis}
\label{convergence_ADMM}
The convergence properties of Algorithm \ref{admm_alg} are presented in this section. We start with two lemmata, which show the monotonicity and boundedness of the augmented Lagrangian function of Problem (\ref{ADMM_auxility}). Two theorems are then given to show that the proposed algorithm converges and that it converges to a stationary point.

The augmented Lagrangian function regarding Problem (\ref{ADMM_auxility}) can be written as
$\mathcal{L}({\bf w}, {\bf v}, {\bf u}) \triangleq \lambda\|{\bf w}\|_{1} + {\bf v}^{\mathrm{H}}{\bf R}_{x}{\bf v} + \frac{\rho}{2}\left(\|{\bf w} - {\bf v} + {\bf u}\|_{2}^{2} - \|{\bf u}\|_{2}^{2} \right)$.
As stated in Section \ref{admm_prob_L1norm}, ${\bf w}$, ${\bf v}$, and ${\bf u}$ are the original, auxiliary, and dual variables, respectively, and $\rho > 0$ is the augmented Lagrangian parameter. In what follows, Lemma \ref{lem_monotonicity} shows that the function value of $\mathcal{L}({\bf w}, {\bf v}, {\bf u})$ is non-increasing, and Lemma \ref{lem_boundedness} shows that the function value $\mathcal{L}({\bf w}, {\bf v}, {\bf u})$ is bounded from below, on the condition that $\rho$ is larger than or equal to a certain value.

%
\begin{lemma}
\label{lem_monotonicity}
As long as the parameter $\rho \geq 2\sqrt{2}\lambda_{\mathrm{max}}({\bf R}_{x})$, the point sequence produces a monotonically non-increasing objective function value sequence $\{\mathcal{L}({\bf w}_{(k)}, {\bf v}_{(k)}, {\bf u}_{(k)})\}$. That is, $\mathcal{L}({\bf w}_{(k + 1)}, {\bf v}_{(k + 1)}, {\bf u}_{(k + 1)}) \leq \mathcal{L}({\bf w}_{(k)}, {\bf v}_{(k)}, {\bf u}_{(k)})$ holds for all $k = 0, 1, 2, \cdots$.
\end{lemma}

\begin{proof}
See Appendix \ref{proof_lem_monotonicity}.
\end{proof}

%
\begin{lemma}
\label{lem_boundedness}
The function value of $\mathcal{L}({\bf w}, {\bf v}, {\bf u})$ is bounded from below by $0$, as long as\footnote{It is worth noting that the lower bound here is not tight, see Appendix \ref{proof_lem_boundedness}.} 
\begin{align}
\label{rho_lowbound}
\rho \geq \frac{2\lambda_{\mathrm{max}}^{2}({\bf R}_{x})}{\lambda_{\mathrm{min}}({\bf R}_{x})}.
\end{align}
\end{lemma}

\begin{proof}
See Appendix \ref{proof_lem_boundedness}.
\end{proof}

With Lemmata \ref{lem_monotonicity} and \ref{lem_boundedness}, we have the following theorem.

\begin{theorem}
\label{theorem_convergence}
As long as the augmented Lagrangian parameter
\begin{align}
\label{rho_condition}
\rho \geq \mathrm{max} \left\{ 2\sqrt{2}\lambda_{\mathrm{max}}({\bf R}_{x}) , \frac{2\lambda_{\mathrm{max}}^{2}({\bf R}_{x})}{\lambda_{\mathrm{min}}({\bf R}_{x})} \right\},
\end{align}
the objective function value sequence $\{\mathcal{L}({\bf w}_{(k)}, {\bf v}_{(k)}, {\bf u}_{(k)})\}$ generated by Algorithm \ref{admm_alg} converges. Furthermore, as $k \rightarrow \infty$, we have ${\bf w}_{(k+1)} = {\bf w}_{(k)}$, ${\bf v}_{(k+1)} = {\bf v}_{(k)}$, ${\bf u}_{(k+1)} = {\bf u}_{(k)}$, and ${\bf w}_{(k)} = {\bf v}_{(k)}$. 
\end{theorem}

\begin{proof}
See Appendix \ref{proof_theorem_convergence}.
\end{proof}

Moreover, the following theorem shows that the limit point of Algorithm \ref{admm_alg} is a stationary point.

\begin{theorem}
\label{theorem_limit_stationary}
Denote the limit point obtained by using Algorithm \ref{admm_alg} as $({\bf w}_{(k+1)}, {\bf v}_{(k+1)}, {\bf u}_{(k+1)})$. Then, it satisfies the KKT conditions of Problem (\ref{ADMM_auxility}), as
\begin{subequations}
\begin{align}
{\bf 0} & = 2{\bf R}_{x}{\bf v}_{(k+1)} - {\bf y}_{(k+1)}, \\
{\bf w}_{(k+1)} & \in \arg\min_{{\bf w}}  \left\{ \!\!
\begin{array}{l}
 \lambda\|{\bf w}\|_{1} + \mu^{\star}(|{\bf w}^{\mathrm{H}}{\bf a}_{0}|^{2} - 1) \\
 + ~\! \Re\{ \langle {\bf y}_{(k+1)} , {\bf w} \! - \! {\bf v}_{(k+1)} \rangle \}
 \end{array}
\!\! \right\}, \\
{\bf w}_{(k+1)} & = {\bf v}_{(k+1)},
\end{align}
\end{subequations}
where ${\bf y}_{(k+1)} = \rho{\bf u}_{(k+1)}$ is the dual variable corresponding to the equality constraint in (\ref{ADMM_st}), and $\mu^{\star}$ denotes the optimal dual variable corresponding to the inequality constraint in (\ref{ADMM_st}). In words, any limit point of Algorithm \ref{admm_alg} is a stationary solution to Problem (\ref{ADMM_auxility}).
\end{theorem}

\begin{proof}
See Appendix \ref{proof_theorem_limit_stationary}.
\end{proof}

\begin{lemma}
\label{lemma_KKT_originalProb}
The limit point obtained by using Algorithm \ref{admm_alg} is a stationary solution to Problem (\ref{L0_norm}), provided that the tuning parameter $\lambda$ is chosen such that the solution of Problem (\ref{ADMM_auxility}) is the same as that of Problem (\ref{L0_norm}). 
\end{lemma}

\begin{proof}
See Appendix \ref{proof_lemma_KKT_originalProb}.
\end{proof}

\section{ADMM with Re-weighted $\ell_{1}$-norm}
\label{ADMM_enhanced_sparsity}
As has been well-documented in the literature, see e.g. \cite{Candes2008}, the iteratively re-weighted $\ell_{1}$-norm penalty has remarkable advantages over the conventional $\ell_{1}$-norm. Therefore, in this section, we propose an improved approach on the basis of Algorithm \ref{admm_alg}, by replacing the $\ell_{1}$-norm regularization in (\ref{L1_norm}) with the re-weighted $\ell_{1}$-norm. That is, Problem (\ref{L1_norm}) is modified as \vspace{-3mm}
\begin{subequations}
\label{L1_norm_enhanced}
\begin{align}
\min_{{\bf w}} ~~~ & {\bf w}^{\mathrm{H}}{\bf R}_{x}{\bf w} + \lambda\| {\bf 1} \oslash (|{\bf g}| + \epsilon) \odot {\bf w}\|_{1}  \\ 
\text{subject to} ~ & |{\bf w}^{\mathrm{H}}{\bf a}_{0}|^{2} \geq 1,
\end{align}
\end{subequations}
where ${\bf g}$ equals ${\bf w}$ obtained from the previous iteration, and $\epsilon > 0$ is a small scalar providing stability and ensuring that a zero-valued component in ${\bf w}$ does not strictly prohibit a non-zero estimate at the next iteration. Note that once the non-zero entries of the solution of Problem (\ref{L1_norm_enhanced}) are identified, their influence is down-weighted in order to allow more sensitivity for identifying the remaining small but non-zero entries \cite{Candes2008}. This results in a better behavior of (\ref{L1_norm_enhanced}) than (\ref{L1_norm}), which will be corroborated in Section \ref{convergenceandsparsitycontrol}.

The ADMM iteration for Problem (\ref{L1_norm_enhanced}) is the same as (\ref{ADMM_iterate}) except for (\ref{ADMM_w}) which should be replaced by
\begin{align}
\!\!\! {\bf w} \! \gets \! \left\{ \!\!\! \begin{array}{l}
\displaystyle\arg\min_{{\bf w}} ~ \lambda\|{\bf 1} \! \oslash \! (|{\bf g}| \! + \! \epsilon) \! \odot \! {\bf w}\|_{1} + \frac{\rho}{2}\|{\bf w} \! - \! {\bf v} \! + \! {\bf u}\|_{2}^{2} \\ 
\text{subject to} ~ |{\bf w}^{\mathrm{H}}{\bf a}_{0}|^{2} \geq 1.
\end{array}
\right.
\end{align}
Accordingly, the result of ${\bf{\bar{w}}}$ in (\ref{w_bar}) now becomes
\begin{align}
{\bf{\bar{w}}} = \text{sign}({\bf v} - {\bf u}) \odot \left(|{\bf v} - {\bf u}| - {(\lambda {\bf 1})} \oslash {[\rho (|{\bf g}| + \epsilon)] }\right)_{{\bold +}},
\end{align}
and the complete ADMM for solving Problem (\ref{L1_norm_enhanced}) is summarized in Algorithm \ref{admm_enhanced_alg}. The convergence property of Algorithm \ref{admm_enhanced_alg}, and the comparison between Algorithms \ref{admm_alg} and \ref{admm_enhanced_alg}, will be presented using simulations in Section \ref{convergenceandsparsitycontrol}.


\begin{algorithm}[t]
	\caption{ADMM for solving Problem (\ref{L1_norm_enhanced})}
	\label{admm_enhanced_alg}
	\textbf{Input~~~~\!:} ${\bf R}_{x} \in \mathbb{C}^{M \times M}$, ${\bf a}_{0} \in \mathbb{C}^{M}$, $\lambda$, $\rho$, $\epsilon$, $k_{\mathrm{max}}$, $\eta$ \\
	\textbf{Output~~\!:} ${\bf{\widehat w}} \in \mathbb{C}^{M}$ \\
	\textbf{Initialize:} ${\bf v}_{(0)} \gets {\bf v}_{\text{init}}$, ${\bf u}_{(0)} \gets {\bf u}_{\text{init}}$, $k \gets 0$
	\begin{algorithmic}[1]
		\While {not converged} \\
		 ~~ ${\bf{\bar{w}}}_{(k+1)} \gets \text{sign}({\bf v}_{(k)} - {\bf u}_{(k)}) ~ \odot $\\  
		 $ \qquad \qquad \qquad \quad \left( |{\bf v}_{(k)} - {\bf u}_{(k)}| - (\lambda {\bf 1}) \! \oslash \! [\rho (|{\bf v}_{(k)}| + \epsilon)] \right)_{{\bold +}} $ \\
		 ~~ ${\bf w}_{(k+1)} \gets {\bf{\bar{w}}}_{(k+1)} + \frac{\left( 1 - |{\bf{\bar{w}}}_{(k+1)}^{\mathrm{H}}{\bf a}_{0}| \right)_{{\bold +}}}{\|{\bf a}_{0}\|_{2}^{2}|{\bf{\bar{w}}}_{(k+1)}^{\mathrm{H}}{\bf a}_{0}|} {\bf a}_{0}{\bf{\bar{w}}}_{(k+1)}^{\mathrm{H}}{\bf a}_{0} $ \\
		 ~~ ${\bf v}_{(k+1)} \gets \rho(2{\bf R}_{x} + \rho{\bf I})^{-1}({\bf w}_{(k+1)} + {\bf u}_{(k)})$ \\
		 ~~ ${\bf u}_{(k+1)} \gets {\bf u}_{(k)} + {\bf w}_{(k+1)} - {\bf v}_{(k+1)}$ \\
		~~ converged $\gets$ $ k \! + \! 1 \geq k_{\mathrm{max}} $ or $ \| {\bf w}_{(k+1)} \! - \! {\bf v}_{(k+1)} \|_{2} \leq \eta $ \\
		 ~~ $k \gets k+1$
		\EndWhile \State \textbf{end while} \\
		${\bf{\widehat w}} \gets {\bf w}_{(k)}$
	\end{algorithmic}
\end{algorithm}

\section{Simulation Results}
\label{simulation}
In this section, numerical examples are conducted to demonstrate the effectiveness of the proposed algorithms, i.e., Algorithms \ref{admm_alg} and \ref{admm_enhanced_alg}. We first examine the behaviors of the algorithms in Section \ref{convergenceandsparsitycontrol}, in terms of convergence property and beamformer weight sparsity control. Then, in Section \ref{computationalcost}, we compare the computational complexity of the proposed algorithms with those of the other state-of-the-art approaches, including SDR, an SDR variant, and SCA, presented in \cite{Mehanna2012}, \cite{Hamza2019}, and \cite{Park2017}, respectively. In Section \ref{beamformingperformance}, we finally test the performance of sparse array beamformers designed by using different strategies, in terms of array beampattern and output SINR.

\subsection{Convergence and Sparsity Control}
\label{convergenceandsparsitycontrol}
\textbf{First example:} A compact ULA consisting of $M = 12$ antenna sensors is utilized, while $T = 100$ snapshots, one SOI from $\theta_{0} = 0^{\circ}$ and $K = 2$ interference signals from $\theta_{1} = -10^{\circ}$ and $\theta_{2} = 10^{\circ}$, respectively, are considered. The signal-to-noise ratio (SNR) and interference-to-noise ratio (INR) are $\text{SNR} = 10 ~ \text{dB}$ and $\text{INR} = 20 ~ \text{dB}$, respectively. The two proposed algorithms, i.e., Algorithms \ref{admm_alg} and \ref{admm_enhanced_alg} are examined, where the parameters are $\epsilon = 10^{-10}$, $k_{\mathrm{max}} = 10^{3}$, $\eta = 10^{-12}$, ${\bf u}_{\text{init}} = {\bf 0}$, and ${\bf v}_{\text{init}}$ is drawn from the complex standard normal distribution. Three scenarios with different values of the tuning parameter $\lambda$ and the augmented Lagrangian parameter $\rho$ are considered. The results of $\mathcal{L}({\bf w}_{(k)}, {\bf v}_{(k)}, {\bf u}_{(k)})$ versus the ADMM iteration index $k$ are given in Fig. \ref{funcvsnumiter}. It is seen that the function value sequence $\{\mathcal{L}({\bf w}_{(k)}, {\bf v}_{(k)}, {\bf u}_{(k)})\}$ of Algorithm \ref{admm_alg} is monotonically non-increasing and bounded from below, which is consistent with the theoretical analyses in Section \ref{convergence_ADMM}. Additionally, we also observe that although the function value sequence $\{\mathcal{L}({\bf w}_{(k)}, {\bf v}_{(k)}, {\bf u}_{(k)})\}$ of Algorithm \ref{admm_enhanced_alg} is not monotonically non-increasing, it converges eventually. Moreover, Algorithm \ref{admm_enhanced_alg} converges faster than Algorithm \ref{admm_alg}.

\begin{figure}[t]
	\vspace*{-6mm}
	\centerline{\includegraphics[width=0.5\textwidth]{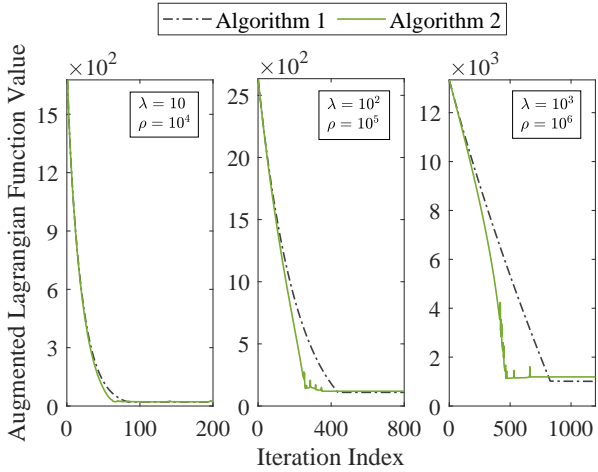}}    
	\caption{Function value of $\mathcal{L}({\bf w}_{(k)}, {\bf v}_{(k)}, {\bf u}_{(k)})$ versus iteration index $k$. (1st example)}
	\label{funcvsnumiter}
\end{figure}

\textbf{Second example:} A compact ULA of $M = 12$ antenna sensors is used, and we wish to select $L$ sensors out of them in the beamformer. We consider 8 situations with different SNR, decreasing from $20 ~ \text{dB}$ to $-15 ~ \text{dB}$ with a stepsize of $5 ~ \text{dB}$. The interference-to-noise ratio is $\text{INR} = 10 ~ \text{dB}$ or $\text{INR} = 20 ~ \text{dB}$, the augmented Lagrangian parameter is $\rho = 2 \times 10^{4}$, and the other parameters are the same as those in the first example. We test the sparsity of the beamformer weight obtained via Algorithms \ref{admm_alg} and \ref{admm_enhanced_alg} w.r.t. the tuning parameter $\lambda$. The curves are averaged over 1000 Monte Carlo runs, and they are displayed in Fig. \ref{sparsityofweight}. It can be seen that, for all 8 situations, any level of sparsity (from 1 to 11) could be attained by both algorithms, and that a larger $\lambda$ produces a smaller sparsity of ${\bf{\widehat{w}}}$, as expected. When SNR increases, the curve becomes more gentle. In addition, the curves by using Algorithm \ref{admm_alg} decrease far more rapidly than those by using Algorithm \ref{admm_enhanced_alg}, which implies that it is much easier to tune $\lambda$ for a specific level of sparsity when Algorithm \ref{admm_enhanced_alg} is employed.

\begin{figure}[t]
	\vspace*{-6mm}
	\centering
	\begin{subfigure}[b]{0.5\textwidth}
		\includegraphics[width=\textwidth]{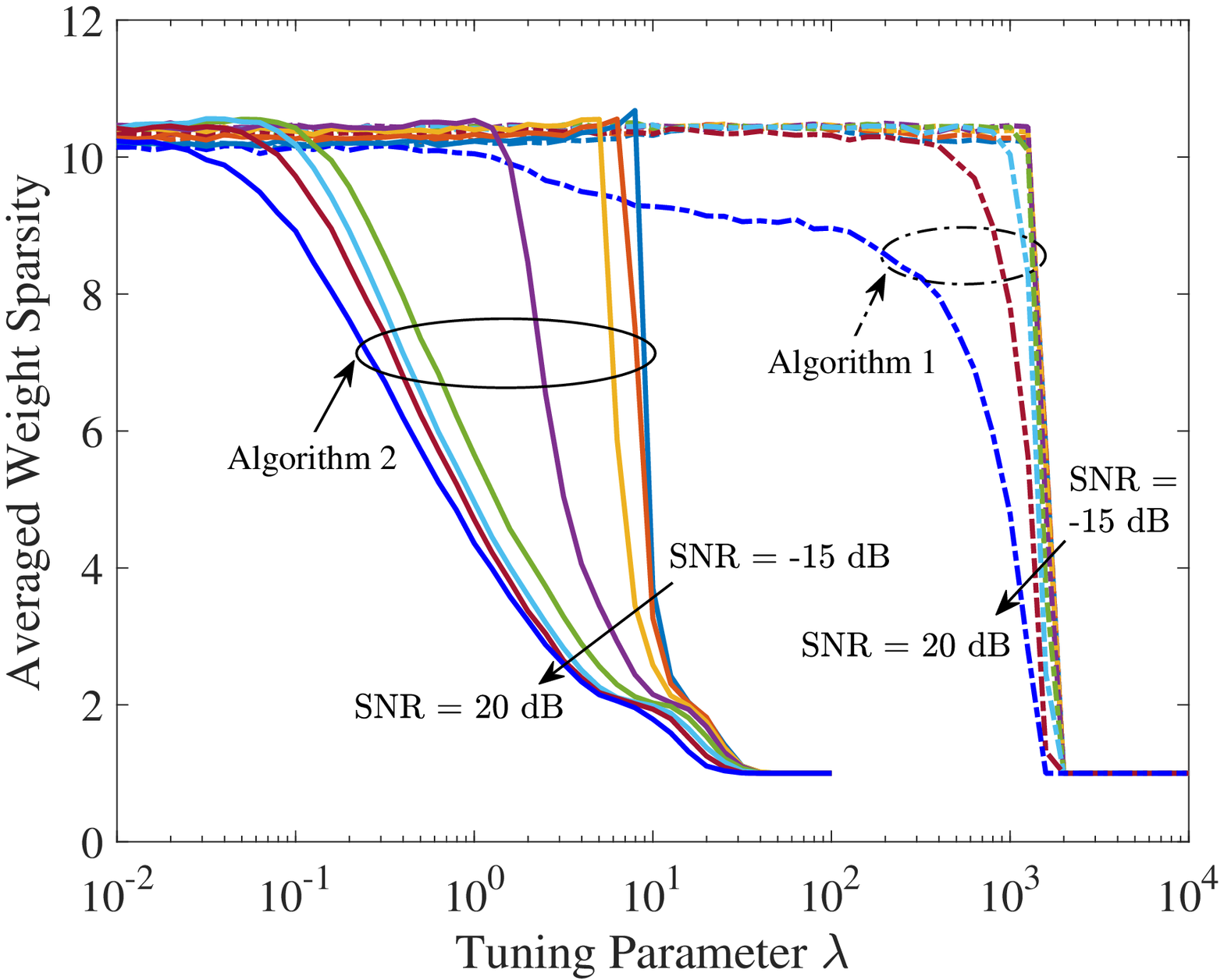}  
		\caption{}
		\label{sparsity_1}
	\end{subfigure}
	~ 
	\begin{subfigure}[b]{0.5\textwidth}
		\includegraphics[width=\textwidth]{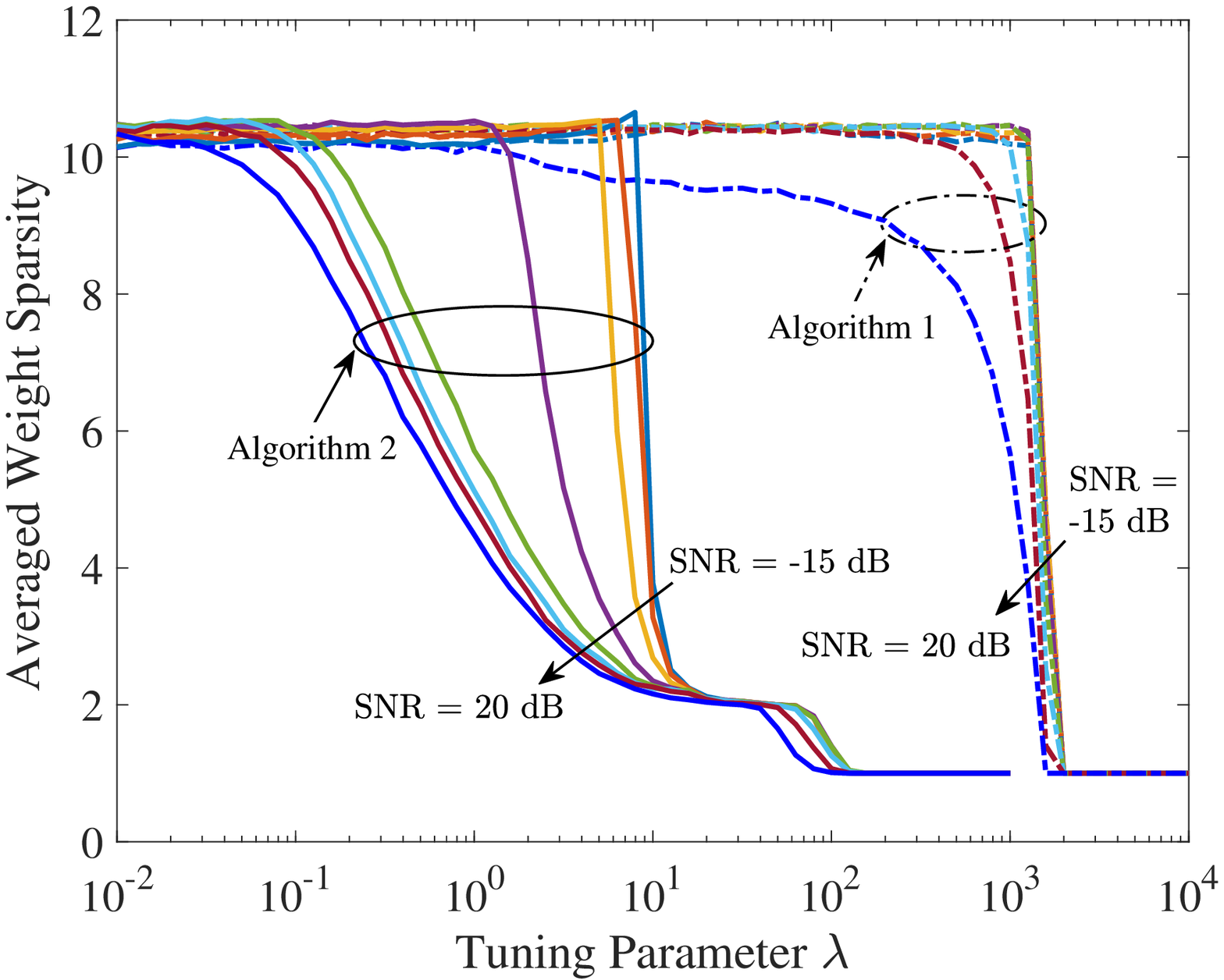}  
		\caption{}
		\label{sparsity_1}
	\end{subfigure}
	\caption{Sparsity of beamformer weight versus tuning parameter $\lambda$ with (a) $\text{INR} = 10 ~ \text{dB}$ and (b) $\text{INR} = 20 ~ \text{dB}$. (2nd example)}
	\label{sparsityofweight}
\end{figure}


Because of the better behavior of Algorithm \ref{admm_enhanced_alg}, rather than both Algorithms \ref{admm_alg} and \ref{admm_enhanced_alg}, we solely consider Algorithm \ref{admm_enhanced_alg} in the remaining simulations, and it will be labelled as ``ADMM''.

\subsection{Computational Complexity}
\label{computationalcost}
We start by rewriting the state-of-the-art methods to suit our problem, i.e., Problem (\ref{L1_norm_enhanced}), using our notations, as \cite{Mehanna2012}
\begin{subequations}
\label{sdr_prob}
\begin{align}
\min_{{\bf W}} ~~~~ & \text{Tr}\{{\bf R}_{x}{\bf W}\} + \lambda\text{Tr}\{{\bf B}  |{\bf W}| \} \\
 \text{subject to} ~ & \left\{ 
\begin{array}{l}
\text{Tr}\{{\bf a}_{0}{\bf a}_{0}^{\mathrm{H}}{\bf W}\} \geq 1, \\
{\bf W} \succeq 0,
\end{array}
\right.
\end{align}
\end{subequations}
and \cite{Hamza2019}
\begin{subequations}
\label{sdrv_prob}
\begin{align}
\min_{{\bf W}, {\bf{\widetilde{W}}}} ~~~~ & \text{Tr}\{{\bf R}_{x}{\bf W}\} + \lambda\text{Tr}\{{\bf B}  {\bf{\widetilde{W}}} \} \\
 \text{subject to} ~ & \left\{ 
\begin{array}{l}
\text{Tr}\{{\bf a}_{0}{\bf a}_{0}^{\mathrm{H}}{\bf W}\} \geq 1, \\
{\bf W} \succeq 0, \\
{\bf{\widetilde{W}}} \geq |{\bf W}|,
\end{array}
\right.
\end{align}
\end{subequations}
where ${\bf B} = {\bf 1}_{M} \oslash (|{\bf G}| + \epsilon)$, with ${\bf G}$ being equal to ${\bf W}$ obtained from the previous iteration. Problems (\ref{sdr_prob}) and (\ref{sdrv_prob}) are referred to as SDR and SDR-V (short for ``SDR Variant''), respectively, in the remaining simulations. If the solution ${\bf W}$ to Problems (\ref{sdr_prob}) and (\ref{sdrv_prob}) is of rank one, their beamformer weight can be calculated as the principal eigenvector of ${\bf W}$; otherwise, extra post-processing based on randomization is required \cite{Sidiropoulos2006}. 

On the other hand, Problem (\ref{L1_norm_enhanced}) can be recast as \cite{Park2017}
\begin{subequations}
\label{sca_prob}
\begin{align}
\min_{{\bf w}} ~~~ & {\bf w}^{\mathrm{H}}{\bf R}_{x}{\bf w} + \lambda\| {\bf b} \odot {\bf w}\|_{1}  \\ 
\text{subject to} ~ & 2\Re\{ {\bf g}^{\mathrm{H}}{\bf a}_{0}{\bf a}_{0}^{\mathrm{H}}{\bf w} \} - |{\bf g}^{\mathrm{H}}{\bf a}_{0}|^{2} \geq 1,
\end{align}
\end{subequations}
where ${\bf b} = {\bf 1} \oslash (|{\bf g}| + \epsilon)$ with $\bf g$ being equal to ${\bf w}$ obtained from the previous iteration, the same as what has been introduced in Section \ref{ADMM_enhanced_sparsity}. Problem (\ref{sca_prob}) is denoted as SCA in the remaining simulations. The key of the success of SCA method lies in the fact that $|{\bf w}^{\mathrm{H}}{\bf a}_{0}|^{2}$ is a convex function w.r.t. ${\bf w}$ and thus $|{\bf w}^{\mathrm{H}}{\bf a}_{0}|^{2} \geq 2\Re\{ {\bf g}^{\mathrm{H}}{\bf a}_{0}{\bf a}_{0}^{\mathrm{H}}{\bf w} \} - |{\bf g}^{\mathrm{H}}{\bf a}_{0}|^{2}$ holds for all ${\bf w}$ and any given (known) ${\bf g}$. 

If a general-purpose SDR solver, such as the interior point method, is adopted to solve Problems (\ref{sdr_prob}) and (\ref{sdrv_prob}), the worst case complexity can be as high as $\mathcal{O}(M^{6.5})$ per iteration \cite{Huang2016}. The cost of solving Problem (\ref{sca_prob}) could be smaller, if further effort is made, for instance, by taking care of the structure of the problem. However, since this is out of the scope of this paper, in our simulations, we simply utilize the CVX toolbox \cite{cvx} to solve the aforementioned three problems. 

The computational costs of the proposed Algorithm \ref{admm_alg} and Algorithm \ref{admm_enhanced_alg} are the same in terms of big O notation. In what follows, we analyse the cost of Algorithm \ref{admm_enhanced_alg} in detail. The cost of computing ${\bf{\bar w}}_{(k+1)}$ (Lines 2 and 3 in Algorithm \ref{admm_enhanced_alg}) is $\mathcal{O}(M)$. The cost of computing ${\bf w}_{(k+1)}$ (Line 4 in Algorithm \ref{admm_enhanced_alg}) is $\mathcal{O}(M)$. The cost of computing ${\bf v}_{(k+1)}$ (Line 5 in Algorithm \ref{admm_enhanced_alg}) is $\mathcal{O}(M^{3}) + \mathcal{O}(M^{2})$, where $\mathcal{O}(M^{3})$ corresponds to the inversion operation, i.e., $(2{\bf R}_{x} + \rho{\bf I})^{-1}$, and $\mathcal{O}(M^{2})$ is from the matrix multiplication. Note that we can cache the result of $(2{\bf R}_{x} + \rho{\bf I})^{-1}$, to save computations in the subsequent iterations. Finally, the cost of computing ${\bf u}_{(k+1)}$ (Line 6 in Algorithm \ref{admm_enhanced_alg}) is $\mathcal{O}(M)$. Therefore, the total computational cost of Algorithm \ref{admm_enhanced_alg} is $\mathcal{O}(M^{3}) + \mathcal{O}(K_{\text{admm}}M^{2})$, where $K_{\text{admm}}$ denotes the number of iterations required by the proposed ADMM. Note that the above analyses are based on a fixed $\lambda$. We may adjust $\lambda$ to ensure $L$ sensors are selected, as mentioned in Remark \ref{remark_L}. Let $K_{\lambda}$ be the number of $\lambda$ values that are considered, then the total complexity cost of the proposed ADMM is $\mathcal{O}(K_{\lambda}M^{3}) + \mathcal{O}(K_{\lambda}K_{\text{admm}}M^{2})$.

It is worth noting that when ${\bf B}$ in (\ref{sdr_prob}) and (\ref{sdrv_prob}) and ${\bf b}$ in (\ref{sca_prob}) are fixed as ${\bf B} = {\bf 1}_{M}$ and ${\bf b} = {\bf 1}$, Problems (\ref{sdr_prob}), (\ref{sdrv_prob}), and (\ref{sca_prob}) reduce to three approaches for solving Problem (\ref{L1_norm}), which correspond to Algorithm \ref{admm_alg}. As has been confirmed by the numerical results in Section \ref{convergenceandsparsitycontrol}, algorithms with re-weighted $\ell_{1}$-norm regularization are more efficient in the sense that they converge faster and are much easier to control the sparsity of solution, compared to their counterparts with conventional $\ell_{1}$-norm regularization. Therefore, only the former group of approaches, i.e., the above-mentioned SDR (\ref{sdr_prob}), SDR-V (\ref{sdrv_prob}), SCA (\ref{sca_prob}), and Algorithm \ref{admm_enhanced_alg}, are considered in the following simulations.

\textbf{Third example:} We wish to choose $L = 4$ out of $M = 12$ antenna sensors from a compact ULA. One SOI from $\theta_{0} = 0^{\circ}$ and $K = 2$ interferers from $\theta_1 = -10^{\circ}$ and $\theta_{2} = 10^{\circ}$, respectively, are considered, while $\text{SNR} = 0 ~ \text{dB}$ and $\text{INR} = 20 ~ \text{dB}$. The number of snapshots $T$ varies uniformly from $10$ to $150$ with a stepsize of $10$. The other parameters for Algorithm \ref{admm_enhanced_alg} are the same as those of the first example, except for $\rho$ which is set to $\rho = 10^{3}$ in this example. The CPU times of the examined approaches are averaged over $100$ Monte Carlo runs, and they are plotted in Fig. \ref{timevssnapshot}. It is seen that their CPU times are almost unchanged when $T$ varies, and that of the ADMM method is around $10^{-1}$ seconds which is about $10^{3}$ times less than those of the SDR, SDR-V, and SCA methods (which take around $10^{2}$ seconds).

\begin{figure}[t]
	\vspace*{-4mm}
	\centerline{\includegraphics[width=0.5\textwidth]{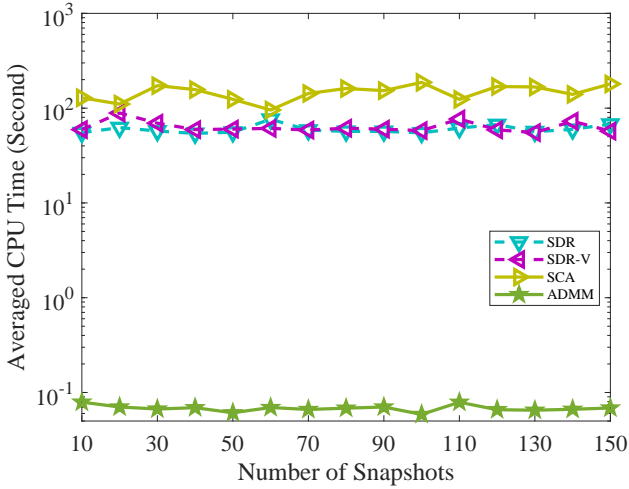}}
	\caption{Averaged CPU time versus $T$. (3rd example)}
	\label{timevssnapshot}
\end{figure}

\textbf{Fourth example:} We wish to select $L = 4$ out of $M$ sensors from a compact ULA. The number of snapshots is fixed as $T = 100$, and the number of sensors $M$ changes from $10$ to $20$. The other parameters remain unchanged as those of the third example. The CPU times of the examined methods are shown in Fig. \ref{timevsnumsensor}, from which it is seen that the CPU times of the SDR, SDR-V, and ADMM methods increase as $M$ increases, while the CPU time of the SCA method keeps almost unchanged when $M$ varies. In addition, the CPU time of the proposed algorithm is much smaller than those of the other three tested approaches. Note that there is a sharp change of the ADMM curve at $M = 15$. This is caused by the increased number of iterations when $M \geq 16$.

\begin{figure}[t]
	\vspace*{-4mm}
	\centerline{\includegraphics[width=0.5\textwidth]{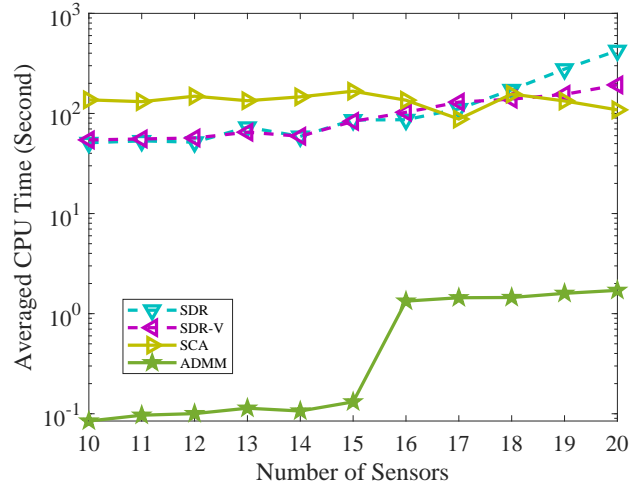}}   
	\caption{Averaged CPU time versus $M$. (4th example)}
	\label{timevsnumsensor}
\end{figure}

\textbf{Fifth example:} We wish to select $L$ out of $M = 12$ sensors from a compact ULA. The number of snapshots is fixed as $T = 100$, and the number of selected sensors $L$ changes from $3$ to $12$. The other parameters remain unchanged as those of the third example. The CPU times of the examined methods are drawn in Fig. \ref{timevsnumselectsensor}, from which it is seen that the CPU times of all the four methods increase as $L$ increases. Besides, the CPU time of the proposed algorithm is shown again much less than those of the other three tested methods.

\begin{figure}[t]
	\vspace*{-6mm}
	\centerline{\includegraphics[width=0.5\textwidth]{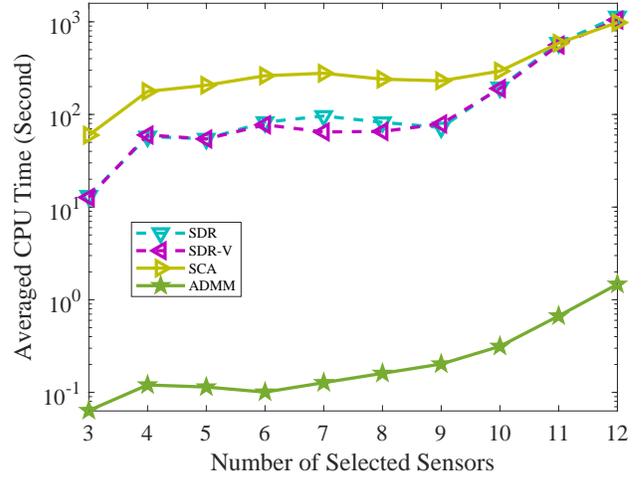}} 
	\caption{Averaged CPU time versus $L$. (5th example)}
	\label{timevsnumselectsensor}
\end{figure}

\subsection{Beamforming Performance}
\label{beamformingperformance}
In the following simulations, we examine the beamforming performance of the proposed algorithm compared with several other sparse array design strategies, in terms of beampattern and output SINR.

\textbf{Sixth example:} We choose $L = 4$ out of $M = 12$ sensors. One SOI from $\theta_{0} = 0^{\circ}$ and $K = 2$ interferers from $\theta_1 = -40^{\circ}$ and $\theta_{2} = 30^{\circ}$, are considered, while $\text{SNR} = 0 ~ \text{dB}$ and $\text{INR} = 20 ~ \text{dB}$. Different sparse array design strategies are examined, including enumeration (i.e., exhaustive search), compact ULA, sparse ULA, random array, nested array \cite{Pal2010}, coprime array \cite{Vaidyanathan2011, Qin2015}, SDR \cite{Mehanna2012}, SDR-V \cite{Hamza2019}, SCA \cite{Park2017}, and the proposed ADMM. The enumeration method tests all possible (i.e., $12 \choose 4$ $= 495$) combinations, and finds out the best and worst cases. The compact ULA includes the first $L = 4$ sensors, the sparse ULA contains the 1st, 3rd, 5th, and 7th sensors, while the random array choose $L = 4$ out of $M = 12$ sensors in a random manner. Besides, the nested array contains the 1st, 2nd, 3rd, and 6th sensors, and coprime array includes the 1st, 3rd, 4th and 5th sensors. SDR, SDR-V, and SCA refer to Problems (\ref{sdr_prob}), (\ref{sdrv_prob}), and (\ref{sca_prob}), respectively. The result of using the whole ULA is also included. Their beampatterns are depicted in Fig. \ref{beampattern}, where we separate them into two subfigures and the ADMM is drawn in both, for a better comparison. From Fig. \ref{beampattern} we observe that the proposed ADMM method provides lower sidelobe and deeper nulls towards the interferences, compared to the others. Their output SINRs in this example are given in TABLE \ref{table_sinr}.



\begin{figure}[t]
	\vspace*{-6mm}
	\centering
	\begin{subfigure}[b]{0.5\textwidth}
		\includegraphics[width=\textwidth]{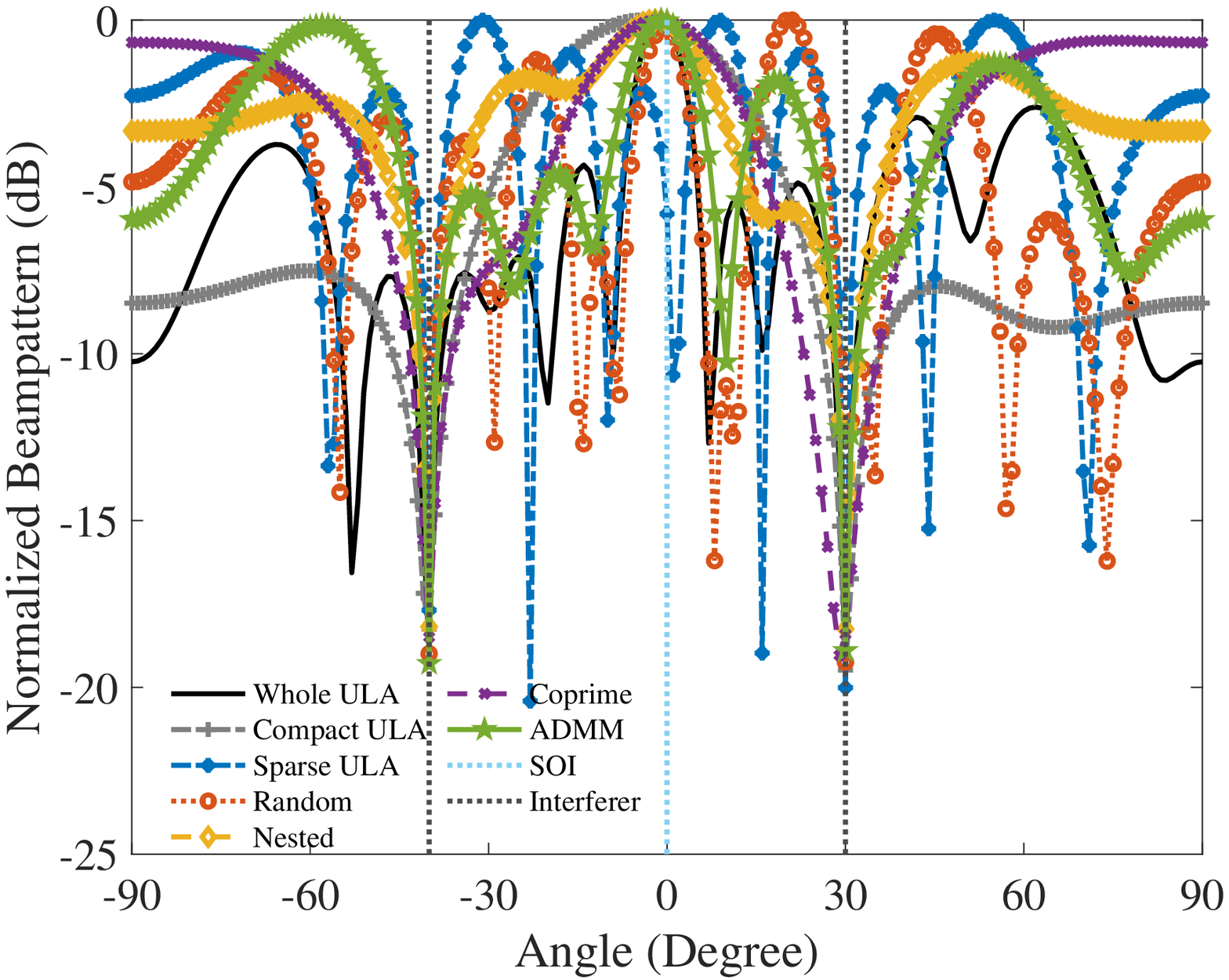}  
		\caption{}
		\label{beampattern_1}
	\end{subfigure}
	~ 
	\begin{subfigure}[b]{0.5\textwidth}
		\includegraphics[width=\textwidth]{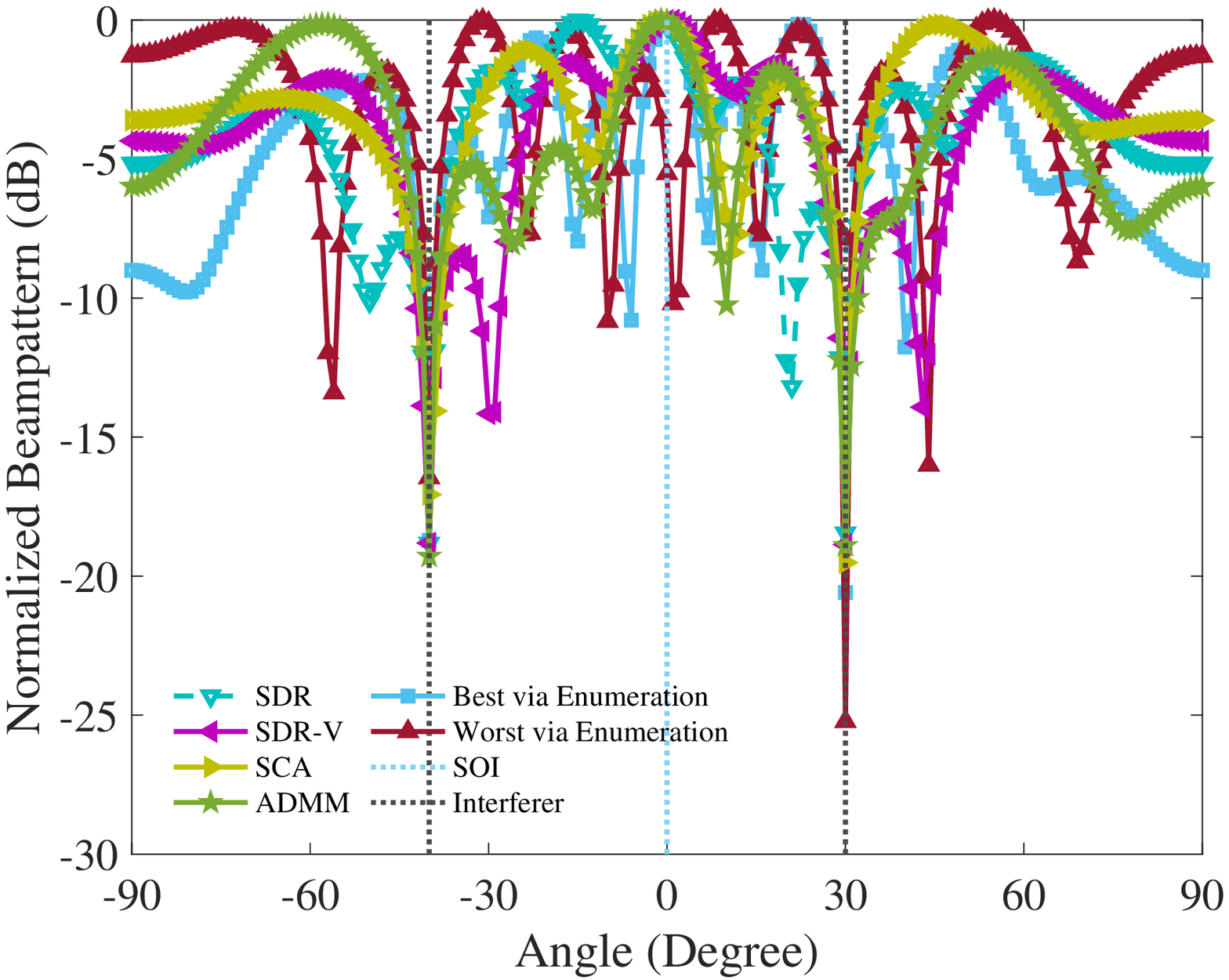}  
		\caption{}
		\label{beampattern_2}
	\end{subfigure}
	\caption{Beampattern comparison with $1$ SOI and $2$ interferers. (6th example)}
	\label{beampattern}
\end{figure}

\begin{table}[t]
\centering
\caption{Output SINR of sixth example.} 
\label{table_sinr}
	\begin{tabular}{|c|c||c|c|}
	\hline
	\textbf{Method} & \textbf{SINR (dB)} & \textbf{Method} & \textbf{SINR (dB)} \\ \hline
	Whole ULA & $7.6160$ & Coprime & $4.2298$ \\ \hline
	Best via Enum. & $5.3874$ & SDR & $4.1792$ \\ \hline
	ADMM & $5.0253$ & Nested & $4.0458$ \\ \hline
	Compact ULA & $4.9967$ & SCA & $3.9229$ \\ \hline
	SDR-V & $4.9670$ & Sparse ULA & $-7.4871$ \\ \hline
	Random & $4.9320$ & Worst via Enum. & $-7.5373$ \\ \hline
	\end{tabular}
\end{table}

\textbf{Seventh example:} We consider the scenario with one SOI whose DOA $\theta_{0}$ changes from $-60^{\circ}$ to $60^{\circ}$, and $K = 2$ interferers from $\theta_{1} = \theta_{0} - 10^{\circ}$ and $\theta_{2} = \theta_{0} + 10^{\circ}$, respectively. The remaining parameters are unchanged as those in the third example. The output SINR versus DOA of the SOI is plotted in Fig. \ref{sinrvsdoaofsoi}. Note that the result of MVDR with whole ULA using the true covariance is termed as ``Optimal with Whole ULA''. It is seen that the proposed method has excellent performance, whose output SINR is less than $0.4 ~\text{dB}$ lower than that of the best case via enumeration, very close (within $0.6 ~\text{dB}$) to the optimal SINR, slightly larger than those of the SDR, SDR-V, and SCA methods, and at least about $2 ~\text{dB}$ larger than those of the other approaches. Two exceptions occur at $\theta_{0} = -55^{\circ}$ and $\theta_{0} = 55^{\circ}$, in which cases the output SINR of ADMM is about $2.5 ~ \text{dB}$ lower than those of the best case via enumeration, SDR, SDR-V, and SCA methods, and is still significantly higher than those of the other approaches. Another interesting result is that the performance of the nested array and the coprime array is even worse than that of the random array. This is because the goal of nested array and coprime array is to make sure more continuous virtual sensors exist in their difference coarray, such that they can estimate more sources than physical sensors. In other words, nested array and coprime array are designed to obtain better performance in DOA estimation, but not necessary to have good performance in beamforming in terms of SINR.

\begin{figure}[t]
	\vspace*{-2mm}
	\centerline{\includegraphics[width=0.5\textwidth]{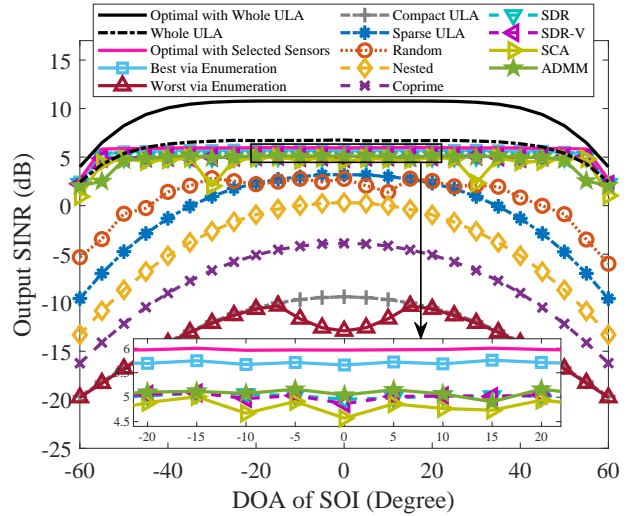}}   
	\caption{Output SINR versus DOA of SOI. (7th example)}
	\label{sinrvsdoaofsoi}
\end{figure}

%
%
%
%

\textbf{Eighth example:} The SNR varies from $-20 ~ \text{dB}$ to $12 ~ \text{dB}$ with a stepsize of $2 ~\text{dB}$. The DOA of the SOI is $\theta_{0} = 0^{\circ}$ and $K = 2$ interference signals come from $\theta_{1} = -10^{\circ}$ and $10^{\circ}$. The other parameters are unchanged compared with those of the previous example. The output SINR versus SNR is depicted in Fig. \ref{sinrvssnr_(1)}. To provide a clearer vision of the results, we calculate the SINR departure of the corresponding methods from the optimal SINR, and draw them in Fig. \ref{departurevssnr}. The figures demonstrate better performance of the proposed scheme than the other sparse array design techniques (except for the best case via enumeration) in terms of output SINR. It is also observed that in the large SNR region, the SINR of the whole array is even smaller than those of sparse arrays. This verifies the statement that the performance of beamforming is affected by not only the beamformer weight, but also the array configuration \cite{Lin1982}. The reason comes from two aspects: On one hand, the performance of MVDR beamformer degrades when SNR is high, since the SOI presents in the training data \cite{Vorobyov2014}. On the other hand, the calculation of SINR for the whole array is different from that for the sparse arrays, that is, the length of the beamformer weight and the size of ${\bf R}_{i+n}$, used for calculating SINR, are different. 


\begin{figure}[t]
	\vspace*{-6mm}
	\centering
	\begin{subfigure}[b]{0.5\textwidth}
		\includegraphics[width=\textwidth]{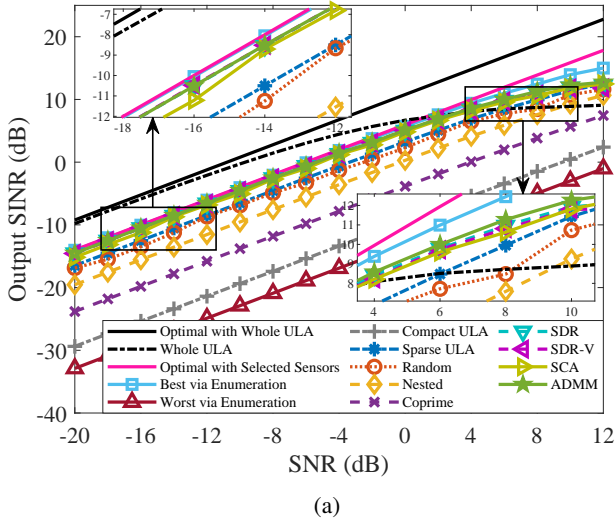}    
		\caption{}
		\label{sinrvssnr_(1)}
	\end{subfigure}
	~ 
	\begin{subfigure}[b]{0.5\textwidth}
		\includegraphics[width=\textwidth]{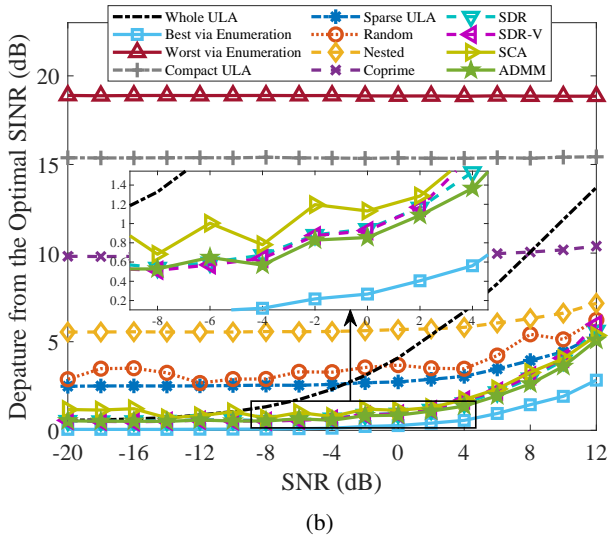}   
		\caption{}
		\label{departurevssnr}
	\end{subfigure}
	\caption{Output SINR versus input SNR. (8th example)}\label{sinrvssnr}
\end{figure}

\textbf{Ninth example:} We examine the output SINR versus the number of snapshots. The simulation setup is the same as that of the third example. The results are shown in Fig. \ref{sinrvssnapshot}, from which we observe that the output SINR of the proposed ADMM is close to that of the best case via enumeration, slightly larger than those of the SDR, SDR-V, and SCA approaches, and significantly larger than those of the others.

\begin{figure}[t]
	\vspace*{-2mm}
	\centering
	\begin{subfigure}[b]{0.5\textwidth}
		\includegraphics[width=\textwidth]{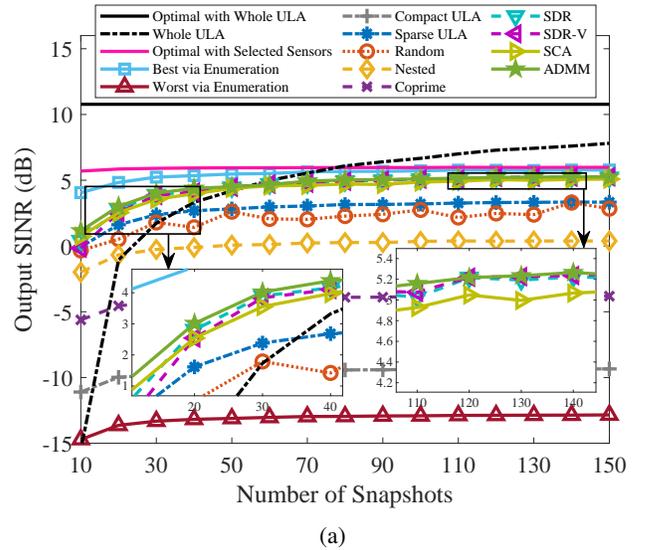}   
		\caption{}
		\label{sinrvssnapshot_(1)}
	\end{subfigure}
	~ 
	\begin{subfigure}[b]{0.5\textwidth}
		\includegraphics[width=\textwidth]{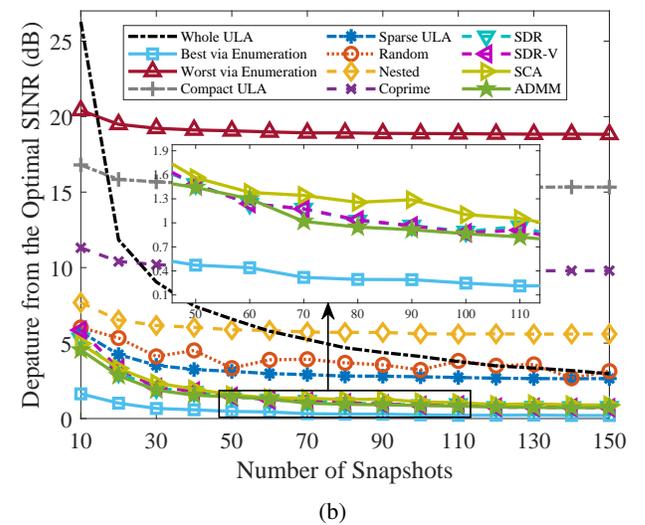}   
		\caption{}
		\label{departurevssnapshot}
	\end{subfigure}
	\caption{Output SINR versus $T$. (9th example)}\label{sinrvssnapshot}
\end{figure}

\textbf{Tenth example:} We examine the output SINR versus the number of sensors. The simulation setup is unchanged as that of the fourth example. The results are plotted in Fig. \ref{sinrvssensor}, from which we again observe that the output SINR of the ADMM is close to that of the best case via enumeration, slightly larger than those of the SDR, SDR-V, and SCA methods, and significantly larger than those of the others.

\begin{figure}[t]
	\vspace*{-3mm}
	\centering
	\begin{subfigure}[b]{0.5\textwidth}
		\includegraphics[width=\textwidth]{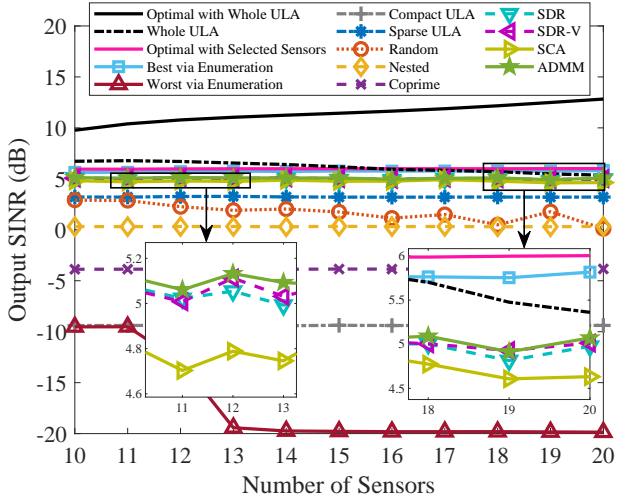}   
		\caption{}
		\label{sinrvssensor_(1)}
	\end{subfigure}
	~ 
	\begin{subfigure}[b]{0.5\textwidth}
		\includegraphics[width=\textwidth]{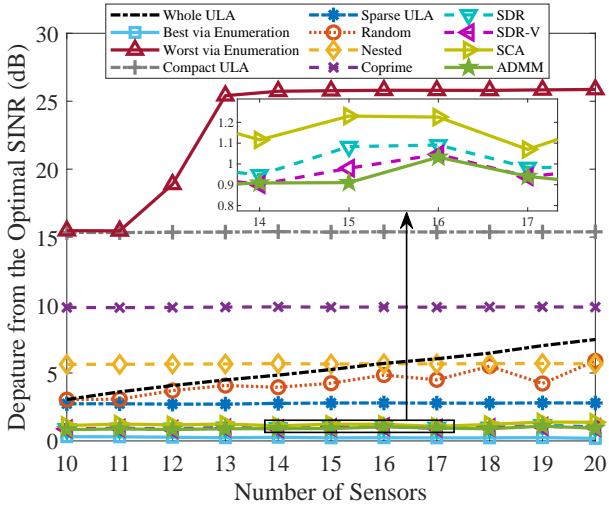}   
		\caption{}
		\label{departurevssensor}
	\end{subfigure}
	\caption{Output SINR versus $M$. (10th example)}\label{sinrvssensor}
\end{figure}

\textbf{Eleventh example:} We examine the output SINR versus the number of selected sensors. The simulation setup is the same as that of the fifth example. Note that, some previous methods are not included in this example since they do not apply to the entire range of $L$. The results are displayed in Fig. \ref{sinrvsselectsensor}. It is seen that when $L$ approaches $M = 12$, the output SINR of all the tested methods converge to one point, because in this case ($L = M$) all the methods select the same sensors, i.e., the whole ULA. On the other hand, when $L < 12$, the proposed ADMM has higher output SINR than the other tested approaches (except for the best case via enumeration).

\begin{figure}[t]
	\vspace*{-2mm}
	\centering
	\begin{subfigure}[b]{0.5\textwidth}
		\includegraphics[width=\textwidth]{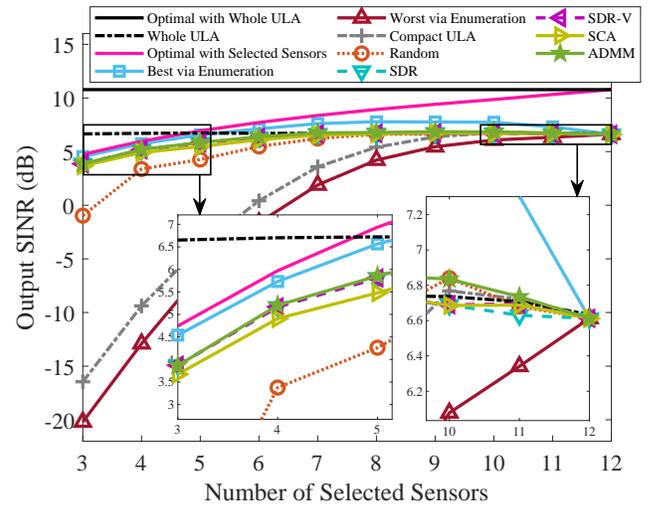}   
		\caption{}
		\label{sinrvsselectsensor_(1)}
	\end{subfigure}
	~ 
	\begin{subfigure}[b]{0.5\textwidth}
		\includegraphics[width=\textwidth]{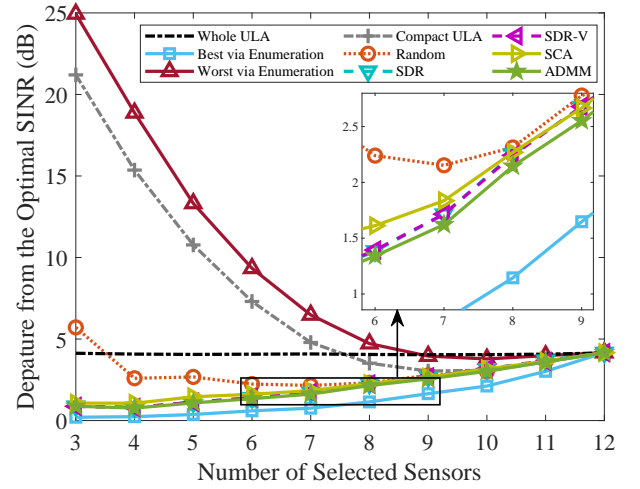}  
		\caption{}
		\label{departurevsselectsensor}
	\end{subfigure}
	\caption{Output SINR versus $L$. (11th example)}\label{sinrvsselectsensor}
\end{figure}

\section{Conclusion and Future Works}
\label{conclusion}
An algorithm based on alternating direction method of multipliers (ADMM) for sparse array beamformer design was proposed. Our approach provides closed-form solutions at each ADMM iteration. Theoretical analyses and numerical simulations were provided to show the convergence of the proposed algorithm. In addition, the algorithm was proved to converge to the set of stationary points. The ADMM algorithm was shown comparable to the exhaustive search method, and slightly better than the state-of-the-art solvers, including the semidefinite relaxation (SDR), an SDR variant (SDR-V), and the successive convex approximation (SCA) methods, and significantly better than several other sparse array design strategies in terms of output signal-to-interference-plus-noise ratio. Moreover, the proposed ADMM algorithm outperformed the SDR, SDR-V, and SCA approaches in terms of computational cost.

Possible future extensions of this work are given as follows.
\begin{itemize}
	\item In this work, we consider narrow-band sources. We might extend it to wide-band sources, as investigated in \cite{Hamza2021}.
	\item We might also consider the situation when the steering vector of SOI  is subject to uncertainties.
	\item We can also extend this work to multiple SOIs. In such a scenario, the problem should be formulated as one with multiple constraints. We can extend the proposed ADMM method to consensus-ADMM method. The results in our previous work \cite{Huang2023} can be applied to this situation.
\end{itemize}

\begin{appendices}

\section{Proof of Equivalence of Problems (\ref{w_har_inequality}) and (\ref{w_har_equality})}
\label{proof_equiv}
The Lagrangian function of Problem (\ref{w_har_inequality}) is given as
\begin{align}
L({\bf w} , \mu) = \|{\bf w} - {\bf{\bar{w}}}\|_{2}^{2} + \mu \left( 1 - |{\bf w}^{\textrm{H}}{\bf a}_{0}|^{2} \right),
\end{align} 
where $\mu \geq 0$ is the Lagrangian dual variable. The optimal solution ${\bf w}^{\star}$ and the optimal dual variable $\mu^{\star}$ satisfy the following identities:
\begin{subequations}
\begin{align}
\label{KKT_w_1}
\left. \frac{\partial L({\bf w} , \mu)}{\partial {\bf w}} \right|_{{\bf w} = {\bf w}^{\star}, \mu = \mu^{\star}} & = 2({\bf w}^{\star} - {\bf{\bar{w}}}) - 2 \mu^{\star} {\bf a}_{0}{\bf a}_{0}^{\textrm{H}}{\bf w}^{\star} = {\bf 0},  \\
\label{KKT_lambda_1}
\left. \frac{\partial L({\bf w} , \mu)}{\partial \mu} \right|_{{\bf w} = {\bf w}^{\star}, \mu = \mu^{\star}} & = 1 - |{\bf w}^{\star \textrm{H}}{\bf a}_{0}|^{2} = 0.
\end{align}
\end{subequations}
From (\ref{KKT_w_1}), we have ${\bf w}^{\star} = ({\bf I} - \mu^{\star}{\bf a}_{0}{\bf a}_{0}^{\textrm{H}})^{-1}{\bf{\bar w}}$. Substituting it into (\ref{KKT_lambda_1}) yields 
\begin{align}
1 - |{\bf{\bar w}}^{\textrm{H}}({\bf I} - \mu^{\star}{\bf a}_{0}{\bf a}_{0}^{\textrm{H}})^{-1}{\bf a}_{0}|^{2} = 0.
\end{align}
If $\mu^{\star} = 0$, the above equation becomes $1 - |{\bf{\bar w}}^{\textrm{H}} {\bf a}_{0}|^{2} = 0$, which contradicts the fact that ${\bf{\bar w}}$ does ${\textbf{not}}$ satisfy $|{\bf{\bar{w}}}^{\mathrm{H}}{\bf a}_{0}|^{2} \geq 1$. Therefore, the optimal dual variable can only be positive, i.e., $\mu^{\star} > 0$.

The complementary slackness condition \cite{Boyd2004_ch5} for Problem (\ref{w_har_inequality}) is 
\begin{align}
\mu^{\star} \geq 0 \text{ ~ and ~ } \mu^{\star} \left( 1 - |{\bf w}^{\star \textrm{H}}{\bf a}_{0}|^{2} \right) = 0.
\end{align} 
As mentioned earlier, we have $\mu^{\star} > 0$, which together with $\mu^{\star} \left( 1 - |{\bf w}^{\star \textrm{H}}{\bf a}_{0}|^{2} \right) = 0$, yields $|{\bf w}^{\star \textrm{H}}{\bf a}_{0}|^{2} = 1$. This indicates that the optimal solution to Problem (\ref{w_har_inequality}) should meet the equality constraint. In other words, Problem (19) is equivalent to Problem (\ref{w_har_equality}).

\section{Proof of Lemma \ref{lem_monotonicity}}
\label{proof_lem_monotonicity}
In order to show that $\mathcal{L}({\bf w}_{(k + 1)}, {\bf v}_{(k + 1)}, {\bf u}_{(k + 1)}) \leq \mathcal{L}({\bf w}_{(k)}, {\bf v}_{(k)}, {\bf u}_{(k)})$ holds $\forall k = 0, 1, 2, \cdots$, where the objective function is defined as $\mathcal{L}({\bf w}, {\bf v}, {\bf u}) \triangleq \lambda\|{\bf w}\|_{1} + {\bf v}^{\mathrm{H}}{\bf R}_{x}{\bf v} + \frac{\rho}{2} \left(\|{\bf w} - {\bf v} + {\bf u}\|_{2}^{2} - \|{\bf u}\|_{2}^{2} \right)$, we formulate their difference as
\begin{subequations}
\label{difference_L}
\begin{align}
& \mathcal{L}({\bf w}_{(k + 1)}, {\bf v}_{(k + 1)}, {\bf u}_{(k + 1)}) - \mathcal{L}({\bf w}_{(k)}, {\bf v}_{(k)}, {\bf u}_{(k)}) \nonumber \\
\label{difference_u}
\!\!\!\!\!\!\! = & ~ [\mathcal{L}({\bf w}_{(k + 1)}, \! {\bf v}_{(k + 1)}, \! {\bf u}_{(k + 1)}) \! - \! \mathcal{L}({\bf w}_{(k + 1)}, \! {\bf v}_{(k + 1)}, \! {\bf u}_{(k)})] \\
\label{difference_wv}
& ~ + [\mathcal{L}({\bf w}_{(k + 1)}, {\bf v}_{(k + 1)}, {\bf u}_{(k)}) - \mathcal{L}({\bf w}_{(k)}, {\bf v}_{(k)}, {\bf u}_{(k)})].
\end{align}
\end{subequations}

In what follows, we separately deal with (\ref{difference_u}) and (\ref{difference_wv}). For (\ref{difference_u}), it is calculated as
\begingroup
\allowdisplaybreaks
\begin{align*}
& \mathcal{L}({\bf w}_{(k + 1)}, {\bf v}_{(k + 1)}, {\bf u}_{(k + 1)}) -  \mathcal{L}({\bf w}_{(k + 1)}, {\bf v}_{(k + 1)}, {\bf u}_{(k)})  \\
\label{diff_u_a}
\stackrel{\text{(a)}}{=} & ~ \frac{\rho}{2}\left(\|{\bf w}_{(k+1)} - {\bf v}_{(k+1)} + {\bf u}_{(k+1)} \|_{2}^{2} - \|{\bf u}_{(k+1)}\|_{2}^{2} \right.  \\
& ~~~~~~~~~~ - \left. \|{\bf w}_{(k+1)} - {\bf v}_{(k+1)} + {\bf u}_{(k)}\|_{2}^{2} + \|{\bf u}_{(k)}\|_{2}^{2} \right) \\
\stackrel{\text{(b)}}{=} & ~ \frac{\rho}{2}\left( \|2{\bf u}_{(k+1)} - {\bf u}_{(k)}\|_{2}^{2} - 2\|{\bf u}_{(k+1)}\|_{2}^{2} + \|{\bf u}_{(k)}\|_{2}^{2} \right) \\
= & ~ \frac{\rho}{2}\left( 2\|{\bf u}_{(k+1)} - {\bf u}_{(k)}\|_{2}^{2} \right) \\
\stackrel{\text{(c)}}{=} & ~ {\rho} \left\| \frac{2}{\rho}{\bf R}_{x}{\bf v}_{(k+1)} - \frac{2}{\rho}{\bf R}_{x}{\bf v}_{(k)} \right\|_{2}^{2}  \\
\stackrel{\text{(d)}}{\leq} & ~ \frac{4}{\rho} \lambda_{\mathrm{max}}^{2}({\bf R}_{x} ) \| {\bf v}_{(k+1)} \! - \! {\bf v}_{(k)} \|_{2}^{2},
\end{align*}
\endgroup
where the definition of $\mathcal{L}({\bf w}, {\bf v}, {\bf u})$ is used in (a); ${\bf w}_{(k+1)} - {\bf v}_{(k+1)} = {\bf u}_{(k+1)} - {\bf u}_{(k)}$ (which is from Line 5 in Algorithm \ref{admm_alg}) has been utilized in (b); ${\bf u}_{(k+1)} = \frac{2}{\rho}{\bf R}_{x}{\bf v}_{(k+1)}$ (which is the result by combining Lines 4 and 5 in Algorithm \ref{admm_alg}) has been employed in (c); and inequality $\| {\bf R}_{x}{\bf v} \|_{2}^{2} = {\bf v}^{\mathrm{H}}{\bf R}_{x}^{\mathrm{H}}{\bf R}_{x}{\bf v} \leq {\bf v}^{\mathrm{H}}[\lambda_{\mathrm{max}}^{2}({\bf R}_{x}){\bf I}]{\bf v} = \lambda_{\mathrm{max}}^{2}({\bf R}_{x})\|{\bf v}\|_{2}^{2}$ has been used in (d).

Now we focus on (\ref{difference_wv}), which can be written as
\begingroup
\allowdisplaybreaks
\begin{align*}
& \mathcal{L}({\bf w}_{(k + 1)}, {\bf v}_{(k + 1)}, {\bf u}_{(k)}) - \mathcal{L}({\bf w}_{(k)}, {\bf v}_{(k)}, {\bf u}_{(k)}) \\
= & [\mathcal{L}({\bf w}_{(k + 1)}, {\bf v}_{(k + 1)}, {\bf u}_{(k)}) - \mathcal{L}({\bf w}_{(k+1)}, {\bf v}_{(k)}, {\bf u}_{(k)})] \\
& + [\mathcal{L}({\bf w}_{(k + 1)}, {\bf v}_{(k)}, {\bf u}_{(k)}) - \mathcal{L}({\bf w}_{(k)}, {\bf v}_{(k)}, {\bf u}_{(k)})] \\
\stackrel{\text{(a)}}{\leq} & \Big[ \Re\{ \! \langle \nabla_{{\bf v}} \mathcal{L}({\bf w}_{(k\!+\!1)}, \! {\bf v}_{(k\!+\!1)}, \! {\bf u}_{(k)}) , {\bf v}_{(k\!+\!1)} \!\! - \! {\bf v}_{(k)} \rangle \! \} \!-\! \frac{\gamma_{{\bf v}}}{2}\|{\bf v}_{(k\!+\!1)} \\ 
& - \! {\bf v}_{(k)}\|_{2}^{2} \Big] +  0  \\ 
\stackrel{\text{(b)}}{=} &  ~\! - \! \frac{\gamma_{{\bf v}}}{2}\|{\bf v}_{(k+1)} \! - \! {\bf v}_{(k)}\|_{2}^{2}   \\  
\stackrel{\text{(c)}}{=} &  ~\! - \! \left[ \! \lambda_{\mathrm{min}}({\bf R}_{x}) \! + \! \frac{\rho }{2} \right] \!\! \|{\bf v}_{(k+1)} \! - \! {\bf v}_{(k)}\|_{2}^{2},   
\end{align*}
\endgroup
where in (a) we used the following two facts: i) $\mathcal{L}({\bf w}, {\bf v}, {\bf u})$ is strongly convex w.r.t. ${\bf v}$ with parameter $\gamma_{{\bf v}} > 0 $ \cite{Ryu2015}, and ii) $\mathcal{L}({\bf w}_{(k + 1)}, {\bf v}_{(k)}, {\bf u}_{(k)}) - \mathcal{L}({\bf w}_{(k)}, {\bf v}_{(k)}, {\bf u}_{(k)}) \leq 0$ since ${\bf w}_{(k+1)}$ is optimal to Problem (\ref{ADMM_w}); in (b) we have used the optimality condition of Problem (\ref{ADMM_v}); in (c) we have used $\gamma_{{\bf v}} = 2\lambda_{\mathrm{min}}({\bf R}_{x}) + \rho$, which is due to the facts that the objective function $\mathcal{L}({\bf w}, {\bf v}, {\bf u})$ is twice continuously differentiable w.r.t. ${\bf v}$, and thus its strong convexity parameter $\gamma_{{\bf v}}$ satisfies $\nabla_{{\bf v}}^{2} \mathcal{L}({\bf w}, {\bf v}, {\bf u}) \succeq \gamma_{{\bf v}}{\bf I}$ for all ${\bf v}$ \cite{Ryu2015}.   

By substituting the above two inequalities back to (\ref{difference_L}), and denoting $\lambda_{\mathrm{max}}({\bf R}_{x})$ and $\lambda_{\mathrm{min}}({\bf R}_{x})$ as $\lambda_{\mathrm{max}}$ and $\lambda_{\mathrm{min}}$ for brevity, respectively, we have
\begin{align*}
\begin{array}{l}
\!\! \mathcal{L}({\bf w}_{(k + 1)}, {\bf v}_{(k + 1)}, {\bf u}_{(k + 1)}) - \mathcal{L}({\bf w}_{(k)}, {\bf v}_{(k)}, {\bf u}_{(k)}) \\
\!\!\! \leq \! \underbrace{ \left[ \! \frac{4}{\rho} \lambda_{\mathrm{max}}^{2} \! - \! \lambda_{\mathrm{min}} \! - \! \frac{\rho }{2} \! \right] \!\! \|{\bf v}_{(k \!+\! 1)} \! - \! {\bf v}_{(k)}\|_{2}^{2}}_{\text{(i)}}. 
\end{array}
\end{align*}
We observe that if $\rho < -\sqrt{\lambda_{\mathrm{min}}^{2} + 8\lambda_{\mathrm{max}}^{2} } - \lambda_{\mathrm{min}}$, which should be deleted as $\rho > 0$, or $\rho > \sqrt{\lambda_{\mathrm{min}}^{2} + 8\lambda_{\mathrm{max}}^{2} } - \lambda_{\mathrm{min}}$, the coefficient $\frac{4}{\rho} \lambda_{\mathrm{max}}^{2} \! - \! \lambda_{\mathrm{min}} \! - \! \frac{\rho }{2} < 0$ and thus $\text{(i)} \leq 0 $. Furthermore, because of $\sqrt{\lambda_{\mathrm{min}}^{2} + 8\lambda_{\mathrm{max}}^{2} } - \lambda_{\mathrm{min}} < \lambda_{\mathrm{min}} + 2\sqrt{2}\lambda_{\mathrm{max}} - \lambda_{\mathrm{min}} = 2\sqrt{2}\lambda_{\mathrm{max}}$, we have the conclusion that as long as $\rho \geq 2\sqrt{2}\lambda_{\mathrm{max}}({\bf R}_{x})$, $ \mathcal{L}({\bf w}_{(k + 1)}, {\bf v}_{(k + 1)}, {\bf u}_{(k + 1)}) - \mathcal{L}({\bf w}_{(k)}, {\bf v}_{(k)}, {\bf u}_{(k)}) \leq \text{(i)} \leq 0$. 

\section{Proof of Lemma \ref{lem_boundedness}}
\label{proof_lem_boundedness}
Note that the augmented Lagrangian function satisfies
\begin{align*}
 & \mathcal{L}({\bf w}, {\bf v}, {\bf u}) \\
 = ~ \! & \lambda\|{\bf w}\|_{1} + {\bf v}^{\mathrm{H}}{\bf R}_{x}{\bf v} + \frac{\rho}{2}\|{\bf w} \!-\! {\bf v} \!+\! {\bf u}\|_{2}^{2} - \frac{\rho}{2}\|{\bf u}\|_{2}^{2} \\
 \stackrel{\text{(a)}}{=} ~ \! & \lambda\|{\bf w}\|_{1} + {\bf v}^{\mathrm{H}}{\bf R}_{x}{\bf v} + \frac{\rho}{2}\|{\bf w} \!-\! {\bf v} \!+\! {\bf u}\|_{2}^{2} - \frac{\rho}{2}\left\|\frac{2}{\rho}{\bf R}_{x}{\bf v}\right\|_{2}^{2} \\
 \stackrel{\text{(b)}}{\geq} ~ \! & \lambda\|{\bf w}\|_{1} + {\bf v}^{\mathrm{H}}{\bf R}_{x}{\bf v} + \frac{\rho}{2}\|{\bf w} \!-\! {\bf v} \!+\! {\bf u}\|_{2}^{2} - \frac{2}{\rho} {\bf v}^{\mathrm{H}}[\lambda_{\mathrm{max}}^{2}({\bf R}_{x}){\bf I}]{\bf v}  \\
 = ~ \! & \lambda\|{\bf w}\|_{1} + {\bf v}^{\mathrm{H}} \! \left( \! {\bf R}_{x} \!-\! \frac{2}{\rho}\lambda_{\mathrm{max}}^{2}({\bf R}_{x}){\bf I} \! \right) \! {\bf v} + \frac{\rho}{2}\|{\bf w} \!-\! {\bf v} \!+\! {\bf u}\|_{2}^{2} \\
 \stackrel{\text{(c)}}{\geq} ~ \! & 0,
 \end{align*}
 where in (a) we have used ${\bf u} = \frac{2}{\rho}{\bf R}_{x}{\bf v}$ $\bigl($which is the result by combining (\ref{ADMM_u}) and (\ref{solution_v})$\bigr)$; in (b) we have used the inequality $\| {\bf R}_{x}{\bf v} \|_{2}^{2} = {\bf v}^{\mathrm{H}}{\bf R}_{x}^{\mathrm{H}}{\bf R}_{x}{\bf v} \leq {\bf v}^{\mathrm{H}}[\lambda_{\mathrm{max}}^{2}({\bf R}_{x}){\bf I}]{\bf v} $; and inequality (c) holds if ${\bf R}_{x} \!-\! \frac{2}{\rho}\lambda_{\mathrm{max}}^{2}({\bf R}_{x}){\bf I} \succeq 0$, which indicates (\ref{rho_lowbound}). Note that inequality (c) is not tight, since the $\ell_{1}$-norm term and the $\ell_2$-norm term could be bounded from below by some large positive value. Therefore, the second term, i.e., the quadratic term w.r.t. ${\bf v}$, has much space to be tuned, which results in the fact that the lower bound for $\rho$ in (\ref{rho_lowbound}) is not tight.

\section{Proof of Theorem \ref{theorem_convergence}}
\label{proof_theorem_convergence}
Denote $\mathcal{L}({\bf w}_{(k + 1)}, {\bf v}_{(k + 1)}, {\bf u}_{(k + 1)})$ and $\mathcal{L}({\bf w}_{(k)}, {\bf v}_{(k)}, {\bf u}_{(k)})$ by $\mathcal{L}_{(k+1)}$ and $\mathcal{L}_{(k)}$, respectively. According to Lemmata \ref{lem_monotonicity} and \ref{lem_boundedness}, the objective function value sequence $\{\mathcal{L}_{(k)}\}$ produced by Algorithm \ref{admm_alg} converges. 
Further, since sequence $\{\mathcal{L}_{(k)}\}$ converges if $\rho$ satisfies (\ref{rho_condition}), we have $\mathcal{L}_{(k+1)} - \mathcal{L}_{(k)} \rightarrow 0$ as $k \rightarrow \infty$. On the other hand, we know from Appendix \ref{proof_lem_monotonicity} that $\mathcal{L}_{(k + 1)} - \mathcal{L}_{(k)} \leq \text{(i)}  \leq 0$, as long as (\ref{rho_condition}) holds, where $\text{(i)}$ is defined in Appendix \ref{proof_lem_monotonicity}. Therefore, when (\ref{rho_condition}) holds and $k \rightarrow \infty$, we have
\begin{align}
0 = \mathcal{L}_{(k+1)} - \mathcal{L}_{(k)} \leq \text{(i)} \leq 0,
\end{align}
meaning that $\text{(i)} = 0$.
This further indicates that 
\begin{align}
\label{wv_k}
 {\bf v}_{(k+1)} = {\bf v}_{(k)}.   
\end{align}
As already mentioned in Appendix \ref{proof_lem_monotonicity}, ${\bf u} = \frac{2}{\rho}{\bf R}_{x}{\bf v}$. By jointly considering ${\bf u} = \frac{2}{\rho}{\bf R}_{x}{\bf v}$ and ${\bf v}_{(k+1)} = {\bf v}_{(k)}$ in (\ref{wv_k}), we obtain
\begin{align}
\label{u_k}
{\bf u}_{(k+1)} = {\bf u}_{(k)}.
\end{align}
Note that the update of ${\bf w}_{(k+1)}$ is based on ${\bf v}_{(k)}$ and ${\bf u}_{(k)}$ (see Lines 2-4 in Algorithm \ref{admm_alg}). This together with (\ref{wv_k}) and (\ref{u_k}) imply that: when $k \rightarrow \infty$, we have
\begin{align}
\label{w_k}
{\bf w}_{(k+1)} = {\bf w}_{(k)}.
\end{align}

Moreover, combining (\ref{u_k}) and ${\bf u}_{(k+1)} = {\bf u}_{(k)} + {\bf w}_{(k+1)} - {\bf v}_{(k+1)}$ (i.e., Line 5 in Algorithm \ref{admm_alg}) yields
\begin{align}
\label{wv_k1}
{\bf w}_{(k+1)} = {\bf v}_{(k+1)},
\end{align}
as $k \rightarrow \infty$. Equivalently, ${\bf w}_{(k)} = {\bf v}_{(k)}$ as $k \rightarrow \infty$.

\section{Proof of Theorem \ref{theorem_limit_stationary}}
\label{proof_theorem_limit_stationary}
The Lagrangian function of (\ref{ADMM_auxility}) is given by
\begin{align}
\label{Lag}
\lambda\|{\bf w}\|_{1} + {\bf v}{\bf R}_{x}{\bf v} + \mu(|{\bf w}^{\mathrm{H}}{\bf a}_{0}| - 1) + \Re\{ \langle {\bf y} , {\bf w} \! - \! {\bf v} \rangle \},
\end{align}
where $\mu$ and ${\bf y}$ are Lagrangian dual variables corresponding to the inequality and equality constraints, respectively. Note that the (Lagrangian) dual variable ${\bf y}$ and the scaled dual variable ${\bf u}$ are related to each other as ${\bf y} = \rho{\bf u}$ \cite{Boyd2011}. A KKT point $({\bf w}^{\star}, {\bf v}^{\star})$ of Problem (\ref{ADMM_auxility}), together with the corresponding dual variables $\mu^{\star}$ and ${\bf y}^{\star}$, satisfies \cite{Hong2016}
\begin{subequations}
\label{KKTconditions}
\begin{align}
\label{KKT_a}
{\bf 0} & = 2{\bf R}_{x}{\bf v}^{\star} - {\bf y}^{\star}, \\
\label{KKT_b}
{\bf w}^{\star} & \in \arg\min_{{\bf w}}  \left\{ \!\!
\begin{array}{l}
 \lambda\|{\bf w}\|_{1} + \mu^{\star}(|{\bf w}^{\mathrm{H}}{\bf a}_{0}|^{2} - 1) \\
 + ~\! \Re\{ \langle {\bf y}^{\star} , {\bf w} \! - \! {\bf v}^{\star} \rangle \}
 \end{array}
\!\! \right\}, \\
\label{KKT_c}
{\bf w}^{\star} & = {\bf v}^{\star}.
\end{align}
\end{subequations}

Our aim is to show any limit point of Algorithm \ref{admm_alg}, referred to as $({\bf w}_{(k+1)}, {\bf v}_{(k+1)}, {\bf u}_{(k+1)})$ or $({\bf w}_{(k+1)}, {\bf v}_{(k+1)}, {\bf y}_{(k+1)}/\rho)$, satisfies (\ref{KKTconditions}). Firstly, note that Line 4 in Algorithm \ref{admm_alg} indicates 
\begin{align}
\label{2Rv}
2{\bf R}_{x}{\bf v}_{(k+1)} - \rho({\bf w}_{(k+1)} - {\bf v}_{(k+1)} + {\bf u}_{(k)}) = {\bf 0}.
\end{align} 
Jointly considering (\ref{u_k}), (\ref{wv_k1}), (\ref{2Rv}), and ${\bf y}_{(k+1)} = \rho{\bf u}_{(k+1)}$ yields $2{\bf R}_{x}{\bf v}_{(k+1)} - {\bf y}_{(k+1)} = {\bf 0}$, which is (\ref{KKT_a}). Additionally, (\ref{wv_k1}) shows that (\ref{KKT_c}) is also achieved. 

Now we turn to (\ref{KKT_b}). According to (\ref{ADMM_w}), we have: ${\bf w}_{(k+1)}$ \vspace{-4mm}
\begingroup
\allowdisplaybreaks
\begin{align*}
= & \arg\min_{{\bf w}} ~ \lambda\|{\bf w}\|_{1} \! + \! \frac{\rho}{2}\|{\bf w} \! - \! {\bf v}_{(k)} \! + \! {\bf u}_{(k)}\|_{2}^{2} \quad \text{s.t.} ~ |{\bf w}^{\mathrm{H}}{\bf a}_{0}|^{2} \geq 1 \\
\stackrel{\text{(a)}}{=} & \arg\min_{{\bf w}} ~ \lambda\|{\bf w}\|_{1} \! + \! \mu^{\star}(|{\bf w}^{\mathrm{H}}{\bf a}_{0}|^{2} \! - \! 1) \! + \! \frac{\rho}{2}\|{\bf w} \! - \! {\bf v}_{(k)} \! + \! {\bf u}_{(k)}\|_{2}^{2} \\
= & \arg\min_{{\bf w}} \left\{ \!\!
\begin{array}{l}
 \lambda\|{\bf w}\|_{1} \! + \! \mu^{\star}(|{\bf w}^{\mathrm{H}}{\bf a}_{0}|^{2} \! - \! 1) \! + \! \frac{\rho}{2}\|{\bf w} - {\bf v}_{(k)}\|_{2}^{2}   \\
 + \frac{\rho}{2}\|{\bf u}_{(k)}\|_{2}^{2} \! + \! \rho\Re\{ \langle {\bf u}_{(k)} , {\bf w} \! - \! {\bf v}_{(k)} \rangle \}
 \end{array}
 \!\! \right\} \\
\stackrel{\text{(b)}}{=} & \arg\min_{{\bf w}} ~ \lambda\|{\bf w}\|_{1} \! + \! \mu^{\star}(|{\bf w}^{\mathrm{H}}{\bf a}_{0}|^{2} \! - \! 1) \! + \! \rho\Re\{ \langle {\bf u}_{(k)} , {\bf w} \! - \! {\bf v}_{(k)} \rangle \} \\
\stackrel{\text{(c)}}{=} & \arg\min_{{\bf w}} ~ \lambda\|{\bf w}\|_{1} \! + \! \mu^{\star}(|{\bf w}^{\mathrm{H}}{\bf a}_{0}|^{2} \! - \! 1) \! + \! \Re\{ \! \langle {\bf y}_{(k+1)} , \! {\bf w} \! - \! {\bf v}_{(k+1)} \rangle \! \},
\end{align*}
\endgroup
where in (a) we have written the constraint into the objective function by involving its optimal dual variable $\mu^{\star}$; in (b) we have utilized the facts that ${\bf w}_{(k+1)} = {\bf w}_{(k)} = {\bf v}_{(k)}$ at any limit point, and $\frac{\rho}{2}\|{\bf u}_{(k)}\|_{2}^{2}$ is a scalar term unrelated to ${\bf w}$; in (c) we have used $\rho{\bf u}_{(k)} = \rho{\bf u}_{(k+1)} = {\bf y}_{(k+1)}$ and ${\bf v}_{(k)} = {\bf v}_{(k+1)}$. This completes the proof of Theorem \ref{theorem_limit_stationary}.

\section{Proof of Lemma \ref{lemma_KKT_originalProb}}
\label{proof_lemma_KKT_originalProb}
First of all, by introducing an auxiliary variable into the original Problem (\ref{L0_norm}), the problem is equivalently rewritten as
\begin{subequations}
\label{L0_norm_res}
\begin{align}
\min_{{\bf w}, {\bf v}} ~ & {\bf v}^{\mathrm{H}}{\bf R}_{x}{\bf v}  \\
\text{subject to} ~ & \left\{ \begin{array}{l}
|{\bf w}^{\mathrm{H}}{\bf a}_{0}|^{2} \geq 1, \\
\|{\bf w}\|_{0} = L, \\
{\bf w} = {\bf v}.
\end{array}  \right.
\end{align}
\end{subequations}
The Lagrangian function of the above problem is given by
\begin{align}
\label{Lag}
{\bf v}{\bf R}_{x}{\bf v} + \mu(|{\bf w}^{\mathrm{H}}{\bf a}_{0}|^{2} - 1) + \nu(\|{\bf w}\|_{0} - L) + \Re\{ \langle {\bf y} , {\bf w} \! - \! {\bf v} \rangle \},
\end{align}
where $\mu$, $\nu$, and ${\bf y}$ are Lagrangian dual variables corresponding to the three constraints, respectively. A KKT point $({\bf w}^{\star}, {\bf v}^{\star})$ of Problem (\ref{L0_norm_res}), together with the corresponding dual variables $\mu^{\star}$, $\nu^{\star}$, and ${\bf y}^{\star}$, satisfies
\begin{subequations}
\label{KKTconditions_res}
\begin{align}
\label{KKT_a_res}
{\bf 0} & = 2{\bf R}_{x}{\bf v}^{\star} - {\bf y}^{\star}, \\
\label{KKT_b_res}
{\bf w}^{\star} & \in \arg\min_{{\bf w}}  \left\{ \!\!
\begin{array}{l}
 \mu^{\star}(|{\bf w}^{\mathrm{H}}{\bf a}_{0}|^{2} - 1) + \nu^{\star}(\|{\bf w}\|_{0} - L) \\
 + ~\! \Re\{ \langle {\bf y}^{\star} , {\bf w} \! - \! {\bf v}^{\star} \rangle \}
 \end{array}
\!\! \right\}, \\
\label{KKT_c_res}
{\bf w}^{\star} & = {\bf v}^{\star}.
\end{align}
\end{subequations}
Our aim is to show any limit point of Algorithm 1, referred to as $({\bf w}_{(k+1)}, {\bf v}_{(k+1)}, {\bf u}_{(k+1)})$, satisfies (\ref{KKTconditions_res}). As has been proven, $({\bf w}_{(k+1)}, {\bf v}_{(k+1)}, {\bf u}_{(k+1)})$ satisfies (\ref{KKTconditions}). Since (\ref{KKT_a_res}) and (\ref{KKT_c_res}) are the same as (\ref{KKT_a}) and (\ref{KKT_c}), respectively. The limit point satisfies (\ref{KKT_a_res}) and (\ref{KKT_c_res}) due to the same reasons presented for (\ref{KKT_a}) and (\ref{KKT_c}), see Appendix E. Our focus is then to show the limit point satisfies (\ref{KKT_b_res}). We have
\begin{align*}
{\bf w}_{(k+1)} \stackrel{\text{(a)}}{=} & \arg\min_{{\bf w}} \left\{ \!\!
\begin{array}{l}
 \lambda\|{\bf w}\|_{1} \! + \! \mu^{\star}(|{\bf w}^{\mathrm{H}}{\bf a}_{0}|^{2} \! - \! 1) \\
 + ~\! \Re\{ \! \langle {\bf y}_{(k+1)} , \! {\bf w} \! - \! {\bf v}_{(k+1)} \rangle \! \} 
 \end{array}
 \!\! \right\}   \\
 \stackrel{\text{(b)}}{=} & \arg\min_{{\bf w}} \! \left\{ \! \mu^{\star}(|{\bf w}^{\mathrm{H}}{\bf a}_{0}|^{2} \! - \! 1) \! + \! \Re\{ \! \langle {\bf y}_{(k+1)} , \! {\bf w} \! - \! {\bf v}_{(k+1)} \rangle \! \} \! \right\} \\ 
 & {\qquad  ~ \text{s.t.}~} \|{\bf w}\|_{0} = L  \\
 \stackrel{\text{(c)}}{=} & \arg\min_{{\bf w}}  \left\{ \!\!
 \begin{array}{l}
 \mu^{\star}(|{\bf w}^{\mathrm{H}}{\bf a}_{0}|^{2} \! - \! 1) + \nu^{\star}(\|{\bf w}\|_{0} - L) \\
 + ~\! \Re\{ \langle {\bf y}^{\star} , {\bf w} \! - \! {\bf v}^{\star} \rangle \}
 \end{array}
 \!\! \right\},
 \end{align*}
 where in (a) we have used the fact that the limit point satisfies (\ref{KKT_b}); in (b) we have used the fact that under some conditions, such as restricted isometry property, and by carefully tuning $\lambda$, the solution of (11) could be the same as that of Problem (\ref{L0_norm}) \cite{Donoho2003, Foucart2013}; and in (c) we have included the constraint into the objective function by involving its optimal dual variable $\nu^{\star}$. The above equation indicates that the limit point satisfies (\ref{KKT_b_res}). This completes the proof of Lemma \ref{lemma_KKT_originalProb}.

\end{appendices}

%

\bibliographystyle{myIEEEtran}       
\bibliography{refs}

%

\end{document}